\documentclass{fmcad}
\usepackage[utf8]{inputenc}
\usepackage[T1]{fontenc}
\usepackage{xspace}
\usepackage{graphicx}

\usepackage{amsmath,amssymb,amsfonts,amsthm}
\usepackage{mathtools}
\newtheorem{assumption}{Assumption}
\newtheorem{definition}{Definition}
\newtheorem{lemma}{Lemma}

\newtheoremstyle{nonitalic} 
  {} 
  {} 
  {\normalfont} 
  {} 
  {\bfseries} 
  {.} 
  { } 
  {} 

\theoremstyle{nonitalic} 
\newtheorem{example}{Example}
\newtheorem{remark}{Remark}

\usepackage{hyperref}
\usepackage{cleveref}
\crefname{assumption}{assumption}{assumptions}
\crefname{definition}{definition}{definitions}
\crefname{example}{example}{examples}
\crefname{lemma}{lemma}{lemmas}
\crefname{remark}{remark}{remarks}

\usepackage{thm-restate}

\usepackage{stmaryrd}
\usepackage{etoolbox}
\usepackage{float}
\usepackage{bm}
\usepackage{multicol}
\usepackage[inline]{enumitem}
\usepackage[normalem]{ulem}
\usepackage{xcolor}

\usepackage{todonotes}

\usepackage{pgf}
\usepackage{pgfplots}
\pgfplotsset{compat=newest}
\usepgfplotslibrary{groupplots}
\usepgfplotslibrary{dateplot}



\usepackage{tabto}
\NumTabs{12}

\usepackage{tikz}
\usetikzlibrary{
    arrows,
    arrows.meta,
    calc,
    decorations.pathreplacing,
    positioning,
    shapes.geometric}

\tikzstyle{process} = [rectangle, minimum width=3cm, minimum height=1cm, text centered, draw=black, fill=orange!30]
\tikzstyle{decision} = [diamond, minimum width=3cm, minimum height=1cm, inner sep = 1pt, text centered, aspect = 3, draw=black, fill=green!30]
\tikzstyle{startstop} = [rectangle, rounded corners, minimum width=1.4cm, minimum height=0.7cm,text centered, draw=black, fill=red!30, inner sep = 1pt]
\tikzstyle{plugin_process} = [rectangle, draw, rounded corners,text width = 6cm, inner sep=6pt, draw=black, fill=purple!30]
\tikzstyle{arrow} = [thick,->,>=stealth]

\usepackage{colortbl}

\usepackage{algorithm}
\usepackage[noend]{algpseudocode}

\algnewcommand\algorithmicfixed{\textbf{Fixed:}}
\algnewcommand\Fixed{\item[\algorithmicfixed]}
\algnewcommand\algorithmicglobal{\textbf{Global:}}
\algnewcommand\Global{\item[\algorithmicglobal]}
\algnewcommand\algorithmicswitch{\textbf{switch}}
\algnewcommand\algorithmiccase{\textbf{case}}
\algnewcommand\algorithmicforeach{\textbf{foreach}}
\algnewcommand\algorithmicnondet{\textbf{nondet}}
\algnewcommand\algorithmicor{\textbf{or}}
\algnewcommand\algorithmicassert{\textbf{assert}}
\algnewcommand\algorithmiclet{\textbf{let}}
\algnewcommand\Assert[1]{\State \algorithmicassert(#1)}
\algdef{SE}[SWITCH]{Switch}{EndSwitch}[1]{\algorithmicswitch\ #1\ \algorithmicdo}{\algorithmicend\ \algorithmicswitch}%
\algtext*{EndSwitch}%
\algdef{SE}[CASE]{Case}{EndCase}[1]{\algorithmiccase\ #1\ \textbf{:}}{\algorithmicend\ \algorithmiccase}%
\algtext*{EndCase}%
\algdef{S}[FOR]{ForEach}[1]{\algorithmicforeach\ #1\ \algorithmicdo}%
\algdef{SE}[NONDET]{Nondet}{EndNondet}{\algorithmicnondet\ }{\algorithmicend\ \algorithmicnondet}%
\algtext*{EndNondet}%
\algdef{SE}[Or]{Or}{EndOr}{\algorithmicor\ }{\algorithmicend\ \algorithmicor}%
\algtext*{EndOr}%
\algdef{SE}[LET]{Let}{EndLet}{\algorithmiclet\ }{\algorithmicend\ \algorithmiclet}%
\algtext*{EndLet}%

\AtBeginEnvironment{algorithm}{\linespread{1.05}}



\usepackage{caption}
\usepackage{diagbox}

\newcommand{\newextmathcommand}[2]{%
    \newcommand{#1}{\ensuremath{#2}\xspace}
}

\newcommand{\renewextmathcommand}[2]{%
    \renewcommand{#1}{\ensuremath{#2}\xspace}
}

\makeatletter
\newcommand{\labeltext}[2]{%
  #1%
  \@bsphack%
  \csname phantomsection\endcsname 
  \def\@currentlabel{#1}{\label{#2}}%
  \@esphack%
}
\makeatother

\makeatletter
               {\list{}{\leftmargin=0pt
                        \labelwidth\z@ \itemindent-\leftmargin
                        }}%
               {\endlist}
\makeatother

\newextmathcommand{\N}{\mathbb{N}}
\newextmathcommand{\Np}{\Nat_+}
\newextmathcommand{\Z}{\mathbb{Z}}
\newextmathcommand{\Q}{\mathbb{Q}}
\newextmathcommand{\Alg}{\overline{\mathbb{Q}}}
\newextmathcommand{\R}{\mathbb{R}}
\newextmathcommand{\C}{\mathbb{C}}
\newextmathcommand{\B}{\mathbb{B}}
\newextmathcommand{\D}{\mathbb{D}}
\newextmathcommand{\FP}{\mathbb{D}_{\textit{fp}}}

\newextmathcommand{\fptime}{\textup{\textsc{FP}}}
\newextmathcommand{\ptime}{\textup{\textsc{P}}}
\newextmathcommand{\np}{\textup{\textsc{NP}}}
\newextmathcommand{\pspace}{\textup{\textsc{PSpace}}}
\newextmathcommand{\nexptime}{\textup{\textsc{NExp}}}
\newextmathcommand{\exptime}{\textup{\textsc{Exp}}}
\newextmathcommand{\twoexptime}{\textup{\textsc{2Exp}}}
\newextmathcommand{\threeexptime}{\textup{\textsc{3Exp}}}
\newextmathcommand{\expspace}{\textup{\textsc{ExpSpace}}}
\newextmathcommand{\twoexpspace}{\textup{\textsc{2ExpSpace}}}
\newextmathcommand{\tower}{\textup{\textsc{Tower}}}
\newextmathcommand{\cclass}{\textup{\textsc{C}}}

\newextmathcommand{\poly}{\textup{poly}}

\renewextmathcommand{\phi}{\varphi}
\renewcommand{\vec}{\bm}

\newextmathcommand{\lcm}{{\rm lcm}}

\newcommand{\infnorm}[1]{\ensuremath{{\parallel}#1{\parallel}}_\infty\xspace}

\DeclarePairedDelimiter{\round}{\lfloor}{\rceil}

\newextmathcommand{\dom}{\textup{dom}}
\newextmathcommand{\vars}{\textup{vars}}

\newcommand{\tool}{\texttt}
\newcommand{\yices}{\tool{Yices2}\xspace}
\newcommand{\mathsat}{\tool{MathSAT}\xspace}
\newcommand{\cvc}{\tool{cvc5}\xspace}
\newcommand{\dreal}{\tool{dReal}\xspace}
\newcommand{\ksmt}{\tool{kSMT}\xspace}
\newcommand{\isat}{\tool{iSAT3}\xspace}
\newcommand{\yicestra}{\tool{Yices-TRA}\xspace}
\newcommand{\yicestraDec}{\tool{Yices-TRA\,(dec $\delta$)}\xspace}

\DeclareFontEncoding{LS2}{}{\noaccents@}
  \DeclareFontSubstitution{LS2}{stix}{m}{n}
  \DeclareSymbolFont{stix@largesymbols}{LS2}{stixex}{m}{n}
  \SetSymbolFont{stix@largesymbols}{bold}{LS2}{stixex}{b}{n}
  \DeclareMathDelimiter{\lBrace}{\mathopen} {stix@largesymbols}{"E8}%
                                            {stix@largesymbols}{"0E}
  \DeclareMathDelimiter{\rBrace}{\mathclose}{stix@largesymbols}{"E9}%
                                            {stix@largesymbols}{"0F}

\newextmathcommand{\defeq}{\coloneqq}
\newextmathcommand{\eqdef}{\defeq}

\newcommand{\sub}[2]{\ensuremath{[#1\,/\,#2]}\xspace}

\definecolor{light-gray}{gray}{0.95}
\newcolumntype{g}{>{\columncolor{light-gray}}r}

\algnewcommand\algorithmicndbranchoutput{\textbf{Output of each branch ($\beta$):}}
\algnewcommand\NDBranchOutput{\item[\algorithmicndbranchoutput]}
\algnewcommand\algorithmicglobalspec{\textbf{Ensuring:}}
\algnewcommand\GlobalSpec{\item[\algorithmicglobalspec]}

\newextmathcommand{\Cmap}{\textup{\textsc{C}}}

\newextmathcommand{\SIGN}{\textup{\textsc{Sign}}}

\newextmathcommand{\cn}{\xi} 
\newextmathcommand{\alg}{\alpha}

\newextmathcommand{\dotb}{\mathbin{\scalebox{0.7}{$\bullet$}}}

\newextmathcommand{\height}{\textup{h}}
\newextmathcommand{\size}{\textup{size}}

\newcommand{\tras}{TRAs\xspace}
\newcommand{\tra}{TRA\xspace}
\newcommand{\nta}{\tra}

\newcommand{\nra}{NRA\xspace}

\newcommand{\mcsat}{\text{MCSAT}\xspace}
\newcommand{\dstep}[2]{{\color{#1}\uline{\color{black}{\texttt{#2}}}}}
\newcommand{\processstep}[1]{\dstep{orange!70}{#1}}
\newcommand{\decisionstep}[1]{\dstep{green!70}{#1}}
\newcommand{\outputstep}[1]{\dstep{red!70}{#1}}
\newextmathcommand{\parmod}{\textup{pm}}
\newextmathcommand{\fal}{\textup{fal}}

\newextmathcommand{\width}{\textup{width}}
\newextmathcommand{\io}{\circ}
\newextmathcommand{\allvars}{\textup{allvars}}

\newextmathcommand{\abox}{\mathcal{B}}

\newcommand{\dsat}{$\delta$\texttt{-SAT}\xspace}

\newcommand{\sat}{\texttt{SAT}\xspace}
\newcommand{\unsat}{\texttt{UNSAT}\xspace}

\newcommand{\defas}{\coloneqq}

\begin{document}
\title{MCSAT Modulo Transcendental Arithmetics}
%
%
\author{
    \IEEEauthorblockN{Jorge Gallego-Hernández~\orcid{0009-0002-2240-1107}}
    \IEEEauthorblockA{\textit{IMDEA Software Institute}\\ \textit{Universidad Politécnica de Madrid}\\
    Madrid, Spain}
    \and
    \IEEEauthorblockN{Enrico Lipparini~\orcid{0009-0009-0428-4403}}
    \IEEEauthorblockA{\textit{University of Cagliari}\\
    Cagliari, Italy}
    \and
    \IEEEauthorblockN{Alessio Mansutti~\orcid{0000-0002-1104-7299}}
    \IEEEauthorblockA{\textit{IMDEA Software Institute}\\
    Madrid, Spain}
}
%
%

%
\maketitle              
%


\begin{abstract}
    We propose a framework for solving quantifier-free formulas from (undecidable) extensions of 
    non-linear real arithmetic (NRA)   with transcendental functions, such as exponential and trigonometric ones.
    The framework extends the Model Constructive Satisfiability calculus (\mcsat), and leverages procedures for NRA and methods from real analysis.

    At its core, our procedure abstracts the input formula to NRA, 
    and lets \mbox{MCSAT} and
    an NRA plugin incrementally build a partial model of the abstracted
    formula. 
    A \emph{Transcendental Real Arithmetic plugin}, acting as an intermediary between MCSAT and the
    NRA plugin, ensures the consistency of the partial model
    and is responsible for refining the~abstracted~formula.


    We implemented our procedure in the Yices2 SMT solver for the sine and exponential functions, and conducted an
    extensive empirical evaluation that shows that our prototype outperforms state-of-the-art solvers on both SAT and UNSAT instances.
\end{abstract}

%

\section{Introduction}\label{sec:introduction}

Many industrial applications in computer science and engineering
reduce to deciding formulas over the real numbers that involve transcendental
functions such as the sine and the exponential~\cite{bartocci2018specification,ratschan2012applications}. Yet, deciding such formulas is notoriously difficult, 
and in many cases outright impossible: 
\begin{itemize}
\item The problem is already undecidable for univariate formulas with 
addition, multiplication, and the sine function~\cite{laczkovich2003removal}; 
and only known to be decidable subject to Schanuel's conjecture for the exponential~\cite{macintyre1996decidability}.
\item Even representing and checking solutions is sometimes not possible, in particular for formulas whose only solutions are transcendental numbers, i.e.,~numbers that are not roots of any polynomial with integer coefficients.
\end{itemize}
These facts motivate the design of principled but incomplete 
methods to tackle problems involving transcendental functions on a best-effort basis.
In this paper, we present such a method. 

Specifically, we address the Satisfiability Modulo Theories (SMT) problem for Transcendental Real Arithmetics (\tras), that is,
deciding the satisfiability of quantifier-free formulas from 
first-order theories extending Non-linear Real Arithmetic (\nra) with 
transcendental functions.
Working in 
the Model Constructive Satisfiability (MCSAT) framework 
(which generalizes ideas from CDCL to the theory level), 
we design a procedure applicable to any such theory whose transcendental functions are total and interval-computable. 

At the high level, our procedure abstracts an input \nta formula to \nra, and uses \nra lemmas to incrementally refine the abstraction. 
A \emph{\mbox{Transcendental Real Arithmetic plugin} \mbox{(\tra plugin)}} acts as an intermediary between the core solver of MCSAT and a (already existing) plugin for NRA.
The role of the \tra plugin is that of ensuring 
that the trail of MCSAT is ``consistent'' in the following sense:
for every literal fully assigned in the trail, its truth value in the \nra abstraction matches the truth value in the \tra concretization.
When the \tra plugin detects an inconsistency, it triggers a conflict,
producing a valid
\nra lemma that is conjoined to the NRA abstraction, effectively refining it.

In \tras, evaluating the truth value of a fully-assigned literal is in general 
an undecidable problem.
We adopt the practical strategy of relaxing the consistency check of the \tra plugin 
up to a tolerance $\delta > 0$.
This means that, in case it is not possible to disprove consistency using over-approximations of the transcendental functions that depend on $\delta$,
the \tra plugin considers the check passed, 
and marks the literal as \emph{$\delta$-consistent}.
The benefit of this approach is that \emph{$\delta$-consistency} is decidable for interval-computable functions, and hence the \tra plugin does not get stuck during its consistency check.
 
The procedure can operate in two modes:
\begin{enumerate}
  \item With tolerance $\delta$ fixed throughout the procedure and provided by the user. In this mode, the procedure may return \texttt{SAT}, \texttt{UNSAT} or
  \dsat. The first two cases carry their usual meaning.
  In the case of \dsat, the procedure has found a solution to the \nra abstraction
  for which all literals were proven consistent or $\delta$-consistent by the TRA plugin, with at least one being $\delta$-consistent.
  \item With tolerance not fixed, and the procedure may only return \texttt{SAT} or \texttt{UNSAT}. A tolerance $\delta$ is still used internally, 
  but rather than answering \dsat, the procedure decreases $\delta$ and backtracks 
  to the empty trail, while preserving all NRA lemmas produced so far.
\end{enumerate}
 
While our procedure works for any \tra in which the transcendental functions are interval-computable, 
we further devote special attention to the case of the sine and exponential functions (plus the constant $\pi$), providing more precise NRA lemmas based on Taylor and Padé approximations.

For $\sin$, $\exp$ and $\pi$,
we implemented a prototype\footnote{Available at \url{https://github.com/arith-lab/yices-tra}.} of our procedure within the \yices SMT solver~\cite{yices2}.
We compared the prototype to
\mathsat~\cite{CimattiGIRS18}, \cvc~\cite{kremer2022cooperating} and
\dreal~\cite{gao2013dreal} on the 2\,512 instances
from~\cite{CimattiGIRS18}. Our experimental 
evaluation shows that our tool outperforms  other tools 
on both \sat and \unsat instances (and, when $\delta$ is fixed, also on \dsat instances).

\section{Preliminaries}\label{section:preliminaries}

In this section we introduce the necessary background material on arithmetic theories over the reals, and provide a high-level overview of the MCSAT framework.
We assume familiarity with propositional satisfiability and the
Conflict-Driven Clause Learning (CDCL) algorithm (a brief overview of CDCL
is given in the appendix).

We recall that the \emph{Satisfiability Modulo Theory (SMT) problem} asks for the satisfiability of a quantifier-free formula in conjunctive normal form (CNF) from a given first-order~theory. 

\subsection{Non-linear Real Arithmetic (NRA)}\label{section:nra}
NRA is formally defined as the first-order theory of the structure ${(\R, 0, 1, +, \cdot, <, =)}$.
This means that, in NRA, variables range over real numbers, and wlog.~atomic formulas are strict inequalities $p(\vec x) < 0$ and equalities $p(\vec x) = 0$, where 
the \emph{term} $p(\vec x)$ is a polynomial with integer coefficients. 
We write $p \leq 0$ and $p \neq 0$ as shortcuts for $\lnot (-p < 0)$
and $\lnot (p = 0)$, 
respectively.
Throughout the paper, we only consider quantifier-free formulas, 
given in CNF. 
The satisfiability problem for NRA is decidable, 
and in practice most decision procedures for this theory are based on Collin's 
Cylindrical Algebraic Decomposition (CAD)~\cite{CollinsCAD}.

A fundamental property of NRA is that if a formula has a solution, then it has an \emph{algebraic solution}~\cite[Theorem~2.81]{BasuPR06}, i.e., a solution assigning algebraic numbers to all variables.
A real number is \emph{algebraic} if it is the root of a univariate polynomial with integer coefficients; else, it is \emph{transcendental}.
For instance, $\sqrt{2}$ is algebraic since it is a solution of the equation $x^2 - 2 = 0$, while $\pi$ is transcendental.
We write~$\Alg$ to denote the set of algebraic numbers.
An algebraic number~$\alpha$ can be effectively represented with a triple $(p,\ell,u)$ 
where $p$ is a univariate polynomial of which~$\alpha$ is a root, and $\ell$ and $u$ are rational numbers such that $\alpha$ is the only root of $p$ in the interval $[\ell,u]$. All standard decision procedures for NRA restrict the search of solution to algebraic solutions.
\subsection{Real Arithmetics with Transcendental Functions}\label{section:tras}

A real-valued function $f \colon \R^n \to \R^m$ is called \emph{transcendental} if its graph cannot be characterized by any formula in NRA; that is, there is no formula $\varphi(\vec x, \vec y)$ in NRA such that, for every $\vec r \in \R^n$ and $\vec s \in \R^m$, $\varphi(\vec r, \vec s)$ holds if and only if $f(\vec r) = \vec s$.
\emph{A Transcendental Real Arithmetic (TRA)} is understood as any first-order theory of a structure that expands NRA with (computable) transcendental functions.
Classical examples of TRAs are given by the first-order theories of the structures ${(\R, 0,1, +, \cdot, \sin, <,=)}$ and ${(\R, 0,1, +, \cdot, \exp, <,=)}$, i.e.,~the expansions of NRA by the sine function and the exponential function $\exp(x) \defas e^x$, respectively.

TRAs are very expressive. For example, even the univariate fragment of ${(\R, 0,1, +, \cdot, \sin, <,=)}$ is undecidable~\cite{laczkovich2003removal}, while~${(\R, 0,1, +, \cdot, \exp, <,=)}$ is only known to be decidable subject to Schanuel's conjecture~\cite{macintyre1996decidability}. Even worse, 
solutions are often
\emph{unverifiable}~\cite[Section~1.3]{Bournez24}: 
there is no general algorithm to decide 
if a given variable assignment satisfies a TRA formula.
This stems from the nature of transcendental functions, whose outputs can typically only be computed to arbitrary precision, whereas the formula may be satisfied only at their exact value.
For these reasons, except for procedures targeting very restricted theories (see, e.g.,~\cite{ChenX23,GallegoM25}), all decision procedures for TRAs are incomplete. 
The procedure we describe in the next section is no exception.
More precisely, our procedure is designed to solve extensions of NRA by 
\emph{interval-computable total}
\emph{functions}, 
defined next.



\begin{definition}[Algebraic box]
    An \emph{algebraic box} in $\R^n$ is the Cartesian product $[a_1, b_1] \times \dots \times [a_n,b_n]$ of $n$ closed intervals of finite length, with algebraic points~$a_i \leq b_i$ as extrema.
    It is represented by the $2n$-tuple $(a_1,b_1,\ldots,a_n,b_n)$.
    We write $\abox_n(\Alg)$ for the set of all algebraic boxes in $\R^n$.

    The \emph{width} of an algebraic box~$B$, denoted $\width(B)$, is defined as $\max_{i=1}^n({b_i - a_i})$. 
    We write $\infnorm{B}$ for the maximum of the $\infty$-norm of the points in $B$, 
    i.e., $\max_{i=1}^n \max(|a_i|, |b_i|)$.
\end{definition} 

\begin{definition}[Interval computable function]\label{def:interval-comput}
    A real-valued function $f \colon \R^n \rightarrow \R^m$ is said to be
    \emph{interval computable} if there is an algorithm~$I_f \colon \abox_n(\Alg) \to \abox_m(\Alg)$ with the following
    property: for every~$\delta > 0$ and $\eta > 0$, there is an $\epsilon > 0$ such that, when given in input any box $B \in \abox_n(\Alg)$ with ${0 < \width(B) < \epsilon}$ and $\infnorm{B} < \eta$,
    the algorithm~$I_f$ returns a box~$C \in \abox_m(\Alg)$ satisfying
    $\width(C) < \delta$ and $f(B) \subseteq C$.
\end{definition}

Three observations are in order. 
First, the composition of two interval computable functions is interval computable. 
Second, addition, multiplication, and many transcendental functions of practical interest are interval computable, including the sine and exponential functions.
For these natural functions, the algorithms $I_f$ are effective and already available in several libraries; 
in particular, our tool uses the \tool{ARB} library~\cite{Arb} for arbitrary-precision interval computation of the sine and exponential functions.
Lastly, any computable constant~$c$ (e.g.,~$\pi$ or $e$) can be treated as a unary interval computable function: given a box of some width $\epsilon$, $I_c$ computes an interval of width $\epsilon$ around $c$.
Throughout the paper, we tacitly see transcendental constants as unary functions applied to the constant $1$.

\begin{example}
    Let us see why the exponential function is interval computable.
    Consider an interval $[-\eta,\eta]$, for some $\eta > 0$,
    and a $\delta > 0$. Let $\epsilon = \frac{\delta}{e^\eta}$.
    Consider a box $[a,b] \subseteq [-\eta,\eta]$ of positive width $b - a < \epsilon$.
    Since $\exp$ is continuous and differentiable, by the mean value theorem ${e^b-e^a = e^\xi(b-a)}$ 
    for some $\xi \in [a,b]$. 
    Since $\exp$ is monotone, we then have $e^b-e^a < e^\eta \epsilon < \delta$.
    On input $[a,b]$, a simple algorithm $I_{\exp}$ consists in computing 
    $p,q \in \Q$ such that $e^a < p < q < e^b$. 
\end{example}

\smallskip
\subsubsection*{$\delta$-satisfiability}
Because of undecidability and unverifiability, symbolic computation often offers very limited support in deciding TRAs, and must be complemented by numerical methods. A principled framework for combining these methods with formal soundness guarantees is given by the notion of 
\mbox{$\delta$-satisfiability} procedure from~\cite{gao2012delta}, which we now recall
(in a slightly extended form that still allows \sat results).

Let $\delta > 0$ be a rational (the \emph{tolerance}). Let~$\phi$ be a quantifier-free CNF formula from some expansion of NRA, with atomic formulas of the form $\tau < 0$ 
or $\tau = 0$. 
The \mbox{$\delta$-weakening} of $\phi$, denoted
$\phi^\delta$, is the formula obtained from $\phi$ by replacing the literals in each clause following the rules: 
    $\tau < 0 \, \rightarrow\, \tau < \delta$,
    $\lnot (\tau < 0) \, \rightarrow\, \lnot (\tau < -\delta)$, \ 
    $\tau = 0 \, \rightarrow\,  -\delta \leq \tau \leq \delta$, and
    $\lnot (\tau = 0) \, \rightarrow\, \top$.
A formula $\phi$ is said to be $\delta$-satisfiable if and only if 
its \mbox{$\delta$-weakening} $\phi^\delta$ is satisfiable. 

\begin{definition}[$\delta$-satisfiability procedure]\label{def:delta-sat-proc}
    A \emph{$\delta$-satisfiability procedure} for an expansion of NRA
    is a procedure that,
    if it terminates on input formula $\phi$, 
    returns
    \texttt{SAT}, \texttt{UNSAT},
    or \dsat. 
    The procedure adheres to the following soundness criterion: if it returns \texttt{SAT} 
    (resp.~\texttt{UNSAT} or \dsat), then $\phi$ 
    must be satisfiable
    (resp.~unsatisfiable or $\delta$-satisfiable). 
    When~$\phi$ is both satisfiable (or unsatisfiable) and \mbox{$\delta$-satisfiable}, the
    procedure can return either of the two answers.
\end{definition}

\subsection{Background on MCSAT}\label{sec:mcsat-tra}
We briefly introduce MCSAT, referring the reader to~\cite{de2013model}
and~\cite{jovanovic2013design} for more detailed 
presentations.


The \mcsat framework consists of a \emph{core solver} that orchestrates a set of
\emph{theory plugins}. Each theory plugin implements a procedure for a specific logical theory. The control flow of the core solver, illustrated
in~\Cref{figure:MCSAT}, closely follows that of CDCL. As in CDCL, the central
data structure is the \emph{trail}, which is a list of four different types of
assignments: 
\begin{itemize}
    \item\emph{Boolean decisions.} Assignments $\ell \mapsto b$ where $\ell$ is
    a literal 
    and $b$ is a Boolean. These are added to the trail during the
    \processstep{Decide} step. 
    \item\emph{Semantic decisions.} Assignments $x \mapsto v$ where $x$ is a
    first-order variable over some sort $S$ (e.g., $\R$), and $v$ is an element of $S$.
    These decisions are also added during the \processstep{Decide} step:
    \mcsat calls the theory plugin associated to $S$ (e.g., the NRA plugin), asking a value for
    $x$. 
    \item\emph{Boolean propagations.} Assignments $\ell \mapsto b$ 
    where $\ell$ is a literal and $b$ is a Boolean, that are added during the \processstep{Propagate} step and
    derived by unit propagation.
    \item\emph{Semantic propagations.} Assignments $\ell \mapsto b$ added during
    the \processstep{Propagate} step by a theory plugin. Specifically, when the
    literal~$\ell$ contains variables already assigned in the trail, the theory
    plugin evaluates its truth value~$b$.
\end{itemize}
\begin{figure}[t]
    \begin{center}
        \scalebox{0.8}{%
        \begin{tikzpicture}[node distance=0.5cm]

            \node (propagate) [process] {Propagate};
            \node (entry) [above left = -0.5cm and 0.5cm of propagate] {input};
            \node (trail_consistent) [decision, below = of propagate, align = center] {Trail  consistent?};
            \node (total_model) [decision, below left = 0.3cm and 1cm of trail_consistent, align = center] {Model complete?};
            \node (valid_backtrack) [decision, below right = 0.3cm and 1cm of trail_consistent, align = center] {Decision in trail?};
            \node (decide) [process, below = of total_model] {Decide};
            \node (analyse_and_backtrack) [process, align = center] at (valid_backtrack |- decide) {Analyse conflict\\ and backtrack};
            \node (sat)    [startstop, below = 0.5cm of trail_consistent] {SAT};
            \node (unsat)  [startstop, below = 0.5cm of sat] {UNSAT};

            \draw [arrow] (propagate)--(trail_consistent);

            \draw [arrow] (trail_consistent) -- node[above,xshift=-2.5mm,yshift=-1.5mm] {yes} (total_model);
            \draw [arrow] (trail_consistent) -- node[above,xshift=2.5mm,yshift=-1.2mm] {no}  (valid_backtrack);

            \draw [arrow] (total_model.south east) -- node[below] {yes} (sat);
            \draw [arrow] (total_model) -- node[left] {no} (decide);

            \draw [arrow] (valid_backtrack) -- node[right] {yes} (analyse_and_backtrack);

            \draw [arrow] (analyse_and_backtrack) -- ++(2cm,0) |- ([yshift=-4pt]propagate.east);

            \draw [arrow] (valid_backtrack) -- node[above] {no} (unsat);
            \draw [arrow] (decide.west) -- ++(-0.8cm,0) |- ([yshift=-6pt]propagate.west);

            \draw [double,arrow] (entry) -- (entry -| propagate.west);
        \end{tikzpicture}}
    \end{center}

    \caption{The MCSAT framework.}\label{figure:MCSAT}
\end{figure}
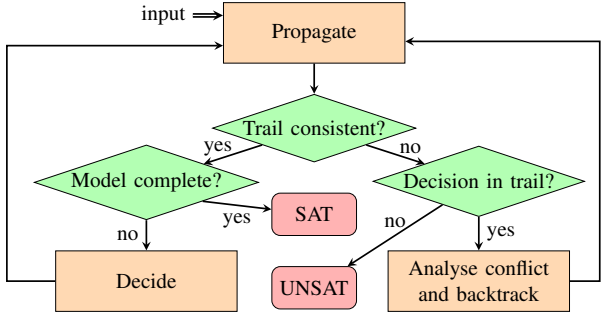%
At each~\processstep{Decide} step, exactly one Boolean or semantic decision is
added to the trail; this decision can only feature variables or literals
that are not yet in the trail. 
The other steps of MCSAT are described below.

\subsubsection*{Step~\processstep{Propagate}}
This step saturates the
trail with all available Boolean and Semantic propagations. Whenever the trail
is updated, \mcsat notifies all plugins of the update. 
The plugins use this
information to refine a \emph{feasibility set} for each variable of their
competence. For example, 
if $x\mapsto 2$ is in the trail, 
and \mcsat notifies the NRA plugin that 
$(x^2-y \leq 0)\mapsto \top$ 
has been added, the plugin
 updates the feasibility set of $y$ to exclude the interval $(-\infty,4]$. The
feasibility sets are used to detect conflicts, as well as to perform semantic decisions.

\smallskip
\subsubsection*{Step~\decisionstep{Consistency}}
A trail is \emph{propositionally consistent} if it does not contain both $\ell
\mapsto \top$ and $\ell \mapsto \bot$, for any literal~$\ell$.
As semantics decisions do not involve variables 
in the trail, a propositionally consistent trail $T$ can be viewed as a map (instead of a list), with $T(t) = v$ if and only if $t \mapsto v$ occurs in~$T$. 
We write $\dom(T)$ for the domain of this map.
A literal is said to be \emph{fully-assigned} in $T$ 
if $\vars(\ell) \subseteq \dom(T)$, where $\vars(\ell)$ is 
the set of first-order variables occurring in~$\ell$.
For a fully-assigned literal $\ell$, we write $T \models \ell$ 
when the assignments of $\vars(\ell)$ in $T$ satisfy the literal.
The notation $T \models \psi$ extends to Boolean combinations~$\psi$ of fully-assigned
literals. 

After \processstep{Propagate}, \mcsat checks the \decisionstep{Consistency} of the trail~$T$ by verifying that it is propositionally consistent, and 
by asking the plugins to check (1)~that every unassigned variable has a non-empty feasibility
set, and (2)~that the value of fully-assigned literals in the trail 
is consistent with their satisfaction. Specifically, 
for each fully-assigned literal $\ell$ with $\ell \mapsto b$ in $T$,
it must hold that $T \models \ell$ if and only if $b = \top$. 

\smallskip
\subsubsection*{Step~\decisionstep{Completeness}}
If the \decisionstep{Consistency} check passes, \mcsat verifies~\decisionstep{Completeness} of the trail, i.e.,
each clause in the input CNF formula 
contains a fully-assigned literal $\ell$ 
such that $\ell \mapsto \top$ or $(\lnot \ell) \mapsto \bot$
is present in the trail
(where $\lnot \ell$ is understood as $\ell'$ if $\ell = \lnot \ell'$). 
If this is the case, \mcsat returns~\outputstep{SAT} since the trail satisfies all clauses in the input formula;
otherwise it moves to the~\processstep{Decide} step.


\smallskip
\subsubsection*{Step~\processstep{Analyse conflict and backtrack}}
When the~\decisionstep{Consistency} check fails, if the trail does not contain a
decision, then the formula is~\outputstep{UNSAT}. Otherwise, \mcsat must
backtrack and learn a new clause $\psi$ that
is implied by the input formula $\phi$ and explains the conflict.
If the conflict stems from Boolean reasoning (e.g., it is due to the Boolean
propagations performed after a Boolean decision), \mcsat behaves as CDCL. 
Otherwise, it invokes the theory plugins. The plugins analyse the trail, and
generate a suitable clause~$\psi$ explaining the conflict, along with the
semantic or Boolean decision causing it. \mcsat conjoins $\psi$ to~$\phi$, and
then backtracks to before the problematic decision.
When backtracking on a semantic decision $x \mapsto v$, \mcsat does not 
necessarily decide $x$ first upon reaching the~\processstep{Decide} step.

\begin{example}
    Consider the following CNF formula of NRA:
    \begin{align*}
        (x^2 = y) &\land (y = -1 \lor y = 2) \land (\lnot(y < 1) \lor \lnot (x < 1)).
    \end{align*}
    Let us discuss a possible run of \mcsat on this formula. The
    initial~\processstep{Propagate} step pushes to the trail~$(x^2 = y) \mapsto
    \top$, as the literal $x^2 = y$ belongs to a unit clause. Next, the control
    flow reaches the~\processstep{Decide} step. Suppose \mcsat 
    performs a Boolean decision, adding $(y = -1) \mapsto \top$ to the trail.
    
    The algorithm now performs the \decisionstep{Consistency} check, which
    fails as the literals $x^2 = y$ and $y = -1$ caused the
    feasibility set of $x$ to become empty. The plugin
    may learn the (tautological) clause $(x^2 \neq y \lor y \geq 0)$, and return the decision $(y = -1) \mapsto \top$ 
    as the cause of the conflict, together with the clause. 
    \mcsat backtracks, updating the trail to 
    $(x^2 = y) \mapsto \top, (y = -1) \mapsto \bot$.
    Unit propagation adds the assignment $(y=2) \mapsto \top$.
    This causes the feasibility set of $y$ to become the singleton $\{2\}$, which in turn causes the feasibility set of $x$ 
    to become $\{\sqrt{2}\}$. 
    Applying~\processstep{Decide} 
    twice adds $y \mapsto 2$ and $x \mapsto \sqrt{2}$ 
    to the trail, making it consistent and complete. 
    \mcsat returns~\outputstep{SAT}. 
\end{example}

\section{TRAs in MCSAT}\label{sec:axioms-tra}\label{sec:tra-plugin}

We describe our MCSAT-based procedure for the first-order theory of a structure 
$(\R, 0, 1, +, \cdot, f_1,\dots,f_k, <, =)$ extending NRA with $k$ interval computable transcendental functions ${f_j \colon \R^{n_j} \to \R}$.
The only assumption we make is to have access, for each $f_j$, to an algorithm $I_{f_j}$ as described in~\Cref{def:interval-comput}.
This general procedure is later tailored to the extension of NRA by the exponential and sine functions in \Cref{section:special-treatment}.

\subsection{High-level view of the procedure}\label{subsec:high-level-view}

As illustrated in \Cref{figure:TRA-plugin},
our procedure starts with a preprocessing phase that constructs 
from the input formula~$\phi$ an abstraction $\phi^\alpha$ in NRA (defined in~\Cref{subsec:first-abstraction}). 
The MCSAT framework is then executed on $\phi^\alpha$, 
with (the core solver of) MCSAT querying the NRA plugin.
The main component of our procedure is given by a \emph{Transcendental Real Arithmetic plugin (TRA plugin)} that  
acts as an intermediary between the core solver and its NRA plugin. 
The TRA plugin has access to both the formula $\phi$ and the abstraction~$\phi^\alpha$ (as well as information to connect the two formulas, discussed~later).

In a nutshell, the TRA plugin is limited to observing the communication between
MCSAT and the NRA plugin, except for when the consistency of the trail must be
checked or conflicts must be analysed.
For instance, if the NRA plugin determines that the trail is consistent, the TRA plugin runs its own
consistency check, but over constraints from~$\phi$ rather than $\phi^{\alpha}$. If the check fails, the TRA plugin proposes
clauses that are implied by~$\phi$ but not by~$\phi^\alpha$, effectively refining
$\phi^\alpha$ enough to trigger a conflict in the NRA~plugin.

\begin{figure}[t]
    \begin{center}
        \scalebox{0.8}{%
        \begin{tikzpicture}[node distance=0.5cm]

            \node (preprocessor) [process,fill=cyan!20] {Preprocessor};
            \node (entry) [above left = -0.5cm and 0.5cm of preprocessor] {input: $\phi$};
            \node (tra-plugin) [process, below = 0.8cm of preprocessor, align = center,fill=cyan!20] {TRA plugin};
            \node (nra-plugin) [process, right = 0.3cm and 1cm of tra-plugin, align = center,fill=cyan!20] {NRA plugin};
            \node (mcsat) [process, left = 0.3cm and 1cm of tra-plugin, align = center,fill=cyan!20] {MCSAT\\[-3pt]\small{(core solver)}};

            \draw [arrow] ([xshift=3mm]preprocessor.south)-- node[right] {$(\phi,\phi^\alpha)$}([xshift=3mm]tra-plugin.north);
            \draw [arrow] ([xshift=-3mm]preprocessor.south) -- ++(0,-0.2) -| node[above, xshift=15mm] {$\phi^\alpha$} (mcsat);
            \draw [double,arrow] (entry) -- (entry -| preprocessor.west);

            \draw [arrow] ([yshift=2mm]tra-plugin.east) -- ([yshift=2mm]nra-plugin.west);
            \draw [arrow] ([yshift=-2mm]nra-plugin.west) -- ([yshift=-2mm]tra-plugin.east);

            \draw [arrow] ([yshift=2mm]mcsat.east) -- ([yshift=2mm]tra-plugin.west);
            \draw [arrow] ([yshift=-2mm]tra-plugin.west) -- ([yshift=-2mm]mcsat.east);

        \end{tikzpicture}}
    \end{center}

    \caption{Schematic view of our procedure.}\label{figure:TRA-plugin}
\end{figure}
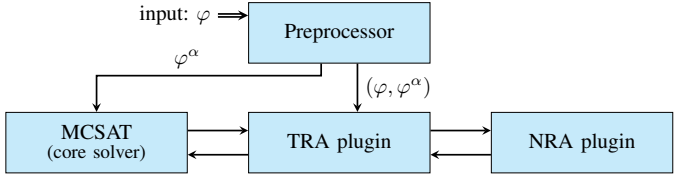

Apart from assuming that MCSAT performs the loop described in the previous section, 
our procedure relies on the following assumptions on the NRA plugin. 
The existence of a plugin satisfying them is guaranteed by the decidability of NRA, 
and the NRA plugin of \yices, which implements the CAD-based procedure from~\cite{jovanovic2013solving}, is a concrete instance.

\begin{assumption}\label{assumption:nra}
    We assume that MCSAT with the NRA plugin is sound, complete and terminating on all instances of NRA.
    We moreover assume: 
    (1) When deciding on a variable, the NRA plugin always returns an algebraic number.
    (2) The NRA plugin can check the consistency of any fully-assigned literal in the trail, when all variables are assigned algebraic numbers.
\end{assumption}

\smallskip
At its core, our procedure respects the specification of \mbox{$\delta$-satisfiability} procedure (\Cref{def:delta-sat-proc}), 
but can also run in a standard
\texttt{SAT}/\texttt{UNSAT} mode (as explained in~\Cref{sec:-delta-vs-standard}).

In practical terms, literals in the trail can 
 take three values: 
$\top$, $\bot$ and $\top^\delta$. 
Except for the TRA plugin,
all other theory plugins treat $\top^\delta$ as $\top$. 
The TRA plugin can weaken literals in the trail, replacing $\top$ with $\top^\delta$ while performing its consistency check (explained in~\Cref{subsection:consistency}).
Roughly speaking, an assignment ${\ell \mapsto \top^\delta}$ in the trail should be understood as ``$\ell$ is a fully-assigned literal in the trail that is true with respect to the abstraction, but whose satisfaction in the original formula is only guaranteed for its $\delta$-weakening''.

In the core solver of MCSAT, the only difference concerns the \decisionstep{Completeness} check. When this check passes (seeing $\top^\delta$ as $\top$), but
a clause in the formula is only true because of a literal with value $\top^\delta$,
MCSAT returns \dsat instead of~\texttt{SAT}.


\subsection{Preprocessing: the first NRA abstraction}\label{subsec:first-abstraction}

We discuss how the first abstraction 
$\phi^\alpha$ is built.
We tag each clause in the input formula $\phi$ as \texttt{original}, 
so that the TRA plugin can distinguish literals from these clauses among those in $\phi^\alpha$. 
These are referred to as \texttt{original} clauses and literals.
Afterwards, we iteratively consider all terms $f(t_1,\dots,t_r)$ 
in $\phi$ starting from the outermost ones, where $f$ is a transcendental function, and perform the rewriting step
\begin{equation*}
    \phi \longrightarrow \phi\sub{f_{(t_1,\dots,t_r)}^{\text{out}}}{f(t_1,\dots,t_r)} \land \textstyle\bigwedge_{i=1}^r f_{t_i}^{\text{in}} = t_i\,,
\end{equation*}
where each $f_{t_i}^{\text{in}}$ and $f_{(t_1,\dots,t_r)}^{\text{out}}$ are fresh real-valued variables, and~$\phi\sub{t'}{t''}$ stands for the formula obtained from
$\phi$ by replacing every occurrence of $t''$ with $t'$. 
This rewriting step is applied until no transcendental function is left in the formula, and the resulting (NRA) formula is $\phi^\alpha$.

The TRA plugin is initialised by providing the formulas $\phi$ and $\phi^\alpha$, as well as maps to reconstruct terms $f(t_1,\dots,t_r)$ from the variables $f_{t_i}^{\text{in}}$ and $f_{(t_1,\dots,t_r)}^{\text{out}}$. Given a literal $\ell$
(resp.~term $\tau$) occurring in $\phi^\alpha$, we write $\ell^\gamma$
(resp.~$\tau^\gamma$) for the literal (resp.~term) obtained by replacing every
$f_{(t_1,\dots,t_r)}^{\text{out}}$ with $f(t_1,\dots,t_r)$ and every $f_{t_i}^{\text{in}}$ with~$t_i$. When the
literal~$\ell$ is \texttt{original}, we write~$\allvars(\ell)$ for all
the variables in~$\ell^\gamma$, plus all variables added when
constructing~$\ell$ from $\ell^\gamma$.

\begin{example}
    Consider $\phi \coloneqq \sin(\cos(x)) + \cos(x)^2 < 0$.
    The NRA abstraction $\phi^\alpha$ is defined as
    $$\sin_{\cos(x)}^{\text{out}} + (\cos_{x}^{\text{out}})^2 < 0 
    \land \sin_{\cos(x)}^{\text{in}} = \cos_{x}^{\text{out}} \land \cos_{x}^{\text{in}} = x.$$
    The (unit) clause $\ell \coloneqq \sin_{\cos(x)}^{\text{out}} + (\cos_{x}^{\text{out}})^2 < 0$ 
    is the only one tagged as \texttt{original}. The set of variables
    $\allvars(\ell)$ is defined as $\{x,\,\sin_{\cos(x)}^{\text{in}},\,\sin_{\cos(x)}^{\text{out}},\,\cos_{x}^{\text{in}},\,\cos_{x}^{\text{out}}\}$. 
\end{example}

The following lemma, whose proof is immediate, 
states fundamental properties of $\phi^\alpha$. 
The TRA plugin preserves these properties as invariants
throughout the procedure, while the NRA abstraction $\phi^\alpha$ is refined.

\begin{restatable}{lemma}{LemmaAbstractionSound}\label{lemma:abstractionsound}
    The formula $\phi^\alpha$ satisfies the following properties.
    (1)~Every variable in $\phi$ occurs in $\phi^\alpha$.
    (2)~Every solution to~$\phi$ can be extended into a solution to~$\phi^\alpha$.
\end{restatable}

After constructing $\phi^\alpha$ and initializing the TRA plugin, MCSAT runs on $\phi^\alpha$ as described in~\Cref{sec:mcsat-tra}.

\begin{remark}\label{remark:initial-propagation}
    MCSAT starts with a \processstep{Propagate} step, 
    during which unit propagation adds to the trail 
    all equalities $f_{t}^{in} = t$ that arise 
    when rewriting $\phi$ into $\phi^{\alpha}$. 
    These equalities remain in the trail throughout 
    the execution of MCSAT.
\end{remark}

We now detail the role of the TRA plugin at each step of the MCSAT framework.





\subsection{Steps: \texorpdfstring{\processstep{Decide}}{Decide} and \texorpdfstring{\processstep{Propagate}}{Propagate}}
In these two steps, the role of the TRA plugin is limited to forwarding data
from MCSAT to the NRA plugin, and returning the results computed by the latter
back to MCSAT. Note that, by~\Cref{assumption:nra}, this implies that all 
semantic decisions assign algebraic values to variables.


\subsection{Step: \texorpdfstring{\decisionstep{Consistency}}{Consistency}}\label{subsection:consistency}

Recall that during this step, the theory plugins verify (1) that the
feasibility sets they maintain for the unassigned variables are non-empty, 
and (2) that the values of fully-assigned literals in the trail 
are consistent with their satisfaction.

The TRA plugin begins this step by forwarding the request to the NRA
plugin, which perform its own consistency check. Should the trail be found 
inconsistent at this stage, the NRA plugin raises a conflict that the TRA
plugin propagates back to MCSAT. Otherwise, the TRA plugin proceeds with 
an additional consistency check.

The additional check does not concern feasibility sets: 
the TRA plugin relies on the feasibility sets of the NRA plugin, which, by
definition of $\phi^\alpha$, are overapproximations of the possible values that
variables can take. While it is sometimes possible to refine these
approximations beyond those of the NRA plugin (see, e.g.,~\cite{ChenX23}), the
undecidability status of even the univariate extension of NRA by the
sine function~\cite{laczkovich2003removal} (which is interval computable) prevents the determination of exact feasibility sets in general.

The consistency check of the TRA plugin is therefore focused solely on 
the second of the points above: verifying that the values of fully-assigned literals in the trail are consistent with their satisfaction. 
Because this verification problem is undecidable (as explained in~\Cref{section:tras}), 
we perform it only ``up to $\delta$''. We call this the \emph{TRA $\delta$-consistency check}:

\begin{definition}\label{def:delta-consistency}
    A \emph{TRA $\delta$-consistency} check is a procedure that given an \texttt{original} literal 
    $\ell$ and  a trail $T$ such that $\allvars(\ell) \subseteq \dom(T)$ returns $\top$, $\bot$ or $\top^\delta$ 
    following the criterion: 
    if it returns $\top$, then $T \models \ell^\gamma$, 
    if it returns $\bot$, then $T \not\models \ell^\gamma$, 
    and if it returns $\top^\delta$, then $T \models (\ell^\gamma)^\delta$.
\end{definition}

\begin{figure}
\begin{algorithm}[H]
    \caption{TRA $\delta$-consistency check.}\label{pseudocode:consistency}
    \begin{algorithmic}[1]
        \Fixed A tolerance $\delta > 0$.
        \Require\begin{minipage}[t]{0.94\linewidth} 
            $(\ell \colon \text{literal}, T \colon \text{trail})$ 
            such that $\ell$ is from an \texttt{original} clause of $\phi^\alpha$, and~$\allvars(\ell) \subseteq \dom(T)$.
            \smallskip
            \end{minipage}
        \Ensure $b \in \{\top,\bot,\top^\delta\}$ such that~if $b = \top$ then $T \models \ell^\gamma$, 
        if $b = \bot$ then $T \not\models \ell^\gamma$, and if $b = \top^\delta$ then $T \models (\ell^\gamma)^\delta$.
        \medskip
        \State \textbf{if} $\ell^\gamma$ contains no transcendental function \textbf{then} \textbf{return} $T(\ell)$
        \Let $\tau \sim 0$  be the atomic formula in $\ell$, where $\sim\, \in \{<,=\}$\label{pseu:cons:tau} 
        \EndLet
        \State \textbf{let} $[i,j] \in \abox_1(\Alg)$ s.t.~$0 < j-i < \delta$ and~$T \models i \leq \tau^\gamma \leq j$\label{pseu:cons:K} 
        \Statex \Comment{$[i,j]$ computed bottom-up by building boxes for all subterms of $\tau^\gamma$, using the algorithms from~\Cref{def:interval-comput}} 
        \State \textbf{return} $b \in \{\top,\bot,\top^\delta\}$ computed using the table:\label{pseu:cons:table}\vspace{3pt}
        \Statex \hfill\scalebox{0.9}{\begin{tabular}{l@{\hspace{0.7em}}c@{\hspace{0.7em}}c@{\hspace{0.7em}}c@{\hspace{0.7em}}c@{\hspace{0.7em}}c}
        \scalebox{0.8}{\diagbox{$\ell$}{$[i,j]$}} & $j < 0,$ & $i > 0,$ & $i = j = 0,$ & $0 = i < j,$ & $i < 0 \leq j$\\\hline\noalign{\vskip 3pt}
        $\tau<0$ & $\top$ &  $\bot$ & $\bot$ & $\bot\phantom{^\delta}$ & $\top^\delta$ \\
        $\tau = 0$ & $\bot$ & $\bot$ & $\top$ & $\top^\delta$ & $\top^\delta$\\
        $\lnot(\tau < 0)$ & $\bot$ & $\top$ & $\top$ & $\top\phantom{^\delta}$ & $\top^\delta$\\
        $\lnot(\tau = 0)$ & $\top$ & $\top$ & $\bot$ & $\top^\delta$ & $\top^\delta$\\
        \end{tabular}}\hfill\mbox{}\label{pseudo:returntable}
    \end{algorithmic}
\end{algorithm}
\end{figure}

\Cref{pseudocode:consistency}, which we discuss in a moment, implements the TRA $\delta$-consistency check. Following~\Cref{def:delta-consistency}, we see that 
there are two cases in which the TRA plugin will not perform this check 
on a literal $\ell$ occurring in the trail:
\begin{itemize}
    \item \textit{When $\ell$ is not {\rm
    \texttt{original}}.} The decision is pragmatic: 
    clauses will only be added to improve the NRA abstraction and therefore 
    soundness is still guaranteed.

    \item \textit{If some variable in the set $\allvars(\ell)$ is unassigned in the trail.} 
        In this case, the TRA
        plugin places $\ell$ in a \emph{watchlist} and defers its consistency
        check until all variables in $\allvars(\ell)$ are assigned
        (note: this must occur for the \decisionstep{Completeness} check to pass). 
\end{itemize}

Consider thus an assignment $\ell \mapsto b'$ in the trail~$T$, such that $\ell$ is \texttt{original}, all variables in $\allvars(\ell)$ are assigned in~$T$, and that passes the consistency check of the NRA plugin, i.e., we have $T \models \ell$ if and only if $b' = \top$.
\Cref{pseudocode:consistency} works as follows. 
First, if $\ell^\gamma$ contains no transcendental function, it returns $b'$. Indeed, in this case $\ell = \ell^\gamma$, and the NRA plugin has already verified $\ell \mapsto b'$. Otherwise, let $\tau$ be the term occurring in $\ell$ (as in line~\ref{pseu:cons:tau}). 
Since interval computable functions are closed under composition, $\tau^\gamma$ is interval computable.
The algorithm computes an algebraic box $[i,j] \in \abox_1(\Alg)$ of width at most $\delta$, such that $T \models i \leq \tau^\gamma \leq j$.
This interval is computed bottom-up, computing boxes $I(\rho)$ for every subterm~$\rho$ of $\tau^\gamma$, starting from the innermost subterms.
Interval computability ensures that by choosing sufficiently small intervals
around the values of variables in $\tau^\gamma$ and the constants $0$ and $1$, the final interval $I(\tau^\gamma) = [i,j]$ will have width less than $\delta$. \Cref{pseudocode:consistency} terminates by performing interval analysis, 
following the table in line~\ref{pseu:cons:table}. 
Note that $\top^\delta$ is only returned in cases where $i$ and $j$ are neither both positive, both negative, nor both zero.

\begin{restatable}{lemma}{TheoremConsistencySpecification}\label{lemma:algo1-is-delta-consistency-check}
    \Cref{pseudocode:consistency} respects its specification. 
\end{restatable}

Returning to the assignment $\ell \mapsto b'$, 
let $b$ denote the output of 
\Cref{pseudocode:consistency} on input $\ell$ and $T$. 
After performing the TRA \mbox{$\delta$-consistency} check, 
the TRA plugin completes its consistency step on $\ell$ 
as follows: 
\begin{algorithmic}[1]
    \If{$b \in \{\top,\bot\}$ and $b \neq b'$}
        \textbf{raise} \texttt{conflict}
    \ElsIf{$b = \top^\delta$ and $b' = \top$}
        \State update $T$: replace $\ell \mapsto \top$ with $\ell \mapsto \top^\delta$
    \EndIf
\end{algorithmic}

\begin{remark}\label{remark:delta-consistency}
    If the above snippet does not raise a \texttt{conflict}, the assignment $\ell \mapsto b''$ (where $b''$ is $b'$ or $\top^\delta$) satisfies the following properties:
    \begin{itemize}
        \item If $b'' = \top$, then $T \models \ell$ and $T \models \ell^\gamma$.
        \item If $b'' = \top^\delta$, then $T \models \ell$ and $T \models (\ell^\gamma)^\delta$.
        \item If $b'' = \bot$, then $T \not\models \ell$.
    \end{itemize}
    Note that this matches the description of $\top^\delta$ given in~\Cref{subsec:high-level-view}: the literal $\ell$ from $\phi^\alpha$ is true,
    while the literal~$\ell^\gamma$ from the original TRA formula $\phi$ need not be true, only its $\delta$-weakening is guaranteed to be.

    One might wonder about the asymmetry between $\top$ and~$\bot$, 
    and whether a more refined treatment of $\bot$ could be given by 
    introducing a value $\bot^\delta$ with semantics 
    ``\,$T \not\models \ell$ and ${T \models (\ell^\gamma)^\delta}$\,''. 
    The short answer is that, since MCSAT is based on propositional logic, 
    introducing $\bot^\delta$ alone does not yield a better procedure: 
    the solver would still need to treat it as $\bot$, except during 
    the~\decisionstep{Completeness} check. At that stage, however,
    no clause can feature only literals assigned to $\bot$ and $\bot^\delta$, 
    since such a situation would have triggered a conflict at the Boolean 
    level during the \decisionstep{Consistency} check.
\end{remark}

%

\subsection{Step: \texorpdfstring{\processstep{Analyse conflict and backtrack}}{Analyse conflict and backtrack}}
\begin{figure}
\begin{algorithm}[H]
    \caption{TRA refinement.}\label{pseudocode:refinement}
    \begin{algorithmic}[1]
        \Require \begin{minipage}[t]{0.94\linewidth} 
                $(\ell \colon \text{literal}, T \colon \text{trail})$ such that $\ell$ is from an \texttt{original} clause of $\phi^\alpha$, 
                $\allvars(\ell) \subseteq \dom(T)$, 
                and \Cref{pseudocode:consistency} on input $(\ell,T)$ returns $b \in \{\top,\bot\}$ with $b \neq T(\ell)$.
                \smallskip
            \end{minipage}
        \Ensure A formula $\psi$ in NRA such that $T \not\models \psi$
        and every solution to $\phi$ can be extended to a solution of $\phi^\alpha\land \psi$.
        \medskip
        \Let $\tau \sim 0$ with $\sim\, \in \{<,=\}$ be the atomic formula in $\ell$\label{pseu:refinement:tau}
        \EndLet
        \State For each subterm $\rho$ of $\tau^\gamma$ ($\tau^\gamma$ included)
        recompute~the~box 
        \Statex \hfill$[i_\rho,j_\rho] \in \abox_1(\Alg)$ used in~\Cref{pseudocode:consistency} 
        to compute $[i,j]$\label{refinement:Isubterms}
        \State $\psi \gets \top$
        \For{$\eta \coloneqq f(\vec t)$ subterm of $\tau^\gamma$ with~$f$ transcendental}\label{pseu:refinement-loop}
            \State \textbf{let} $\mathcal{T}$ be the set of terms in the vector $\vec t$
            \State $\psi \gets \psi \land \big((\bigwedge_{\rho \in \mathcal{T}} i_{\rho} \leq f_{\rho}^{\text{in}} \leq j_{\rho}) \Rightarrow i_\eta \leq f_{\vec t}^\text{out} \leq j_\eta \big)$\label{refinement:psi}
        \EndFor
        \State \textbf{return} $\psi$\label{refinement:return}
    \end{algorithmic}
\end{algorithm}
\end{figure}

Suppose the TRA plugin raises a conflict, 
causing MCSAT to enter the~\processstep{Analyse conflict and backtrack} step 
(when the trail to contain a decision) 
and query the TRA plugin in return.
If the conflict originates from the NRA plugin, the
TRA plugin lets the NRA plugin resolve it, again acting as an intermediary.
Otherwise, the conflict arose following the TRA \mbox{$\delta$-consistency} check. 
In particular, there is an assignment ${\ell \mapsto b'}$ in 
the trail $T$ with the following properties: 
\begin{itemize}
    \item $\ell$ is \texttt{original}, and $\allvars(\ell) \subseteq \dom(T)$.
    \item $T \models \ell$, but $T \not\models \ell^\gamma$. 
    More precisely,~\Cref{pseudocode:consistency} returned $b \in \{\bot,\top\}$ 
    on input $\ell$ and $T$, with $b \neq b'$.
\end{itemize}
The TRA plugin uses $\ell$ in order to refine $\phi^\alpha$ so that it continues to satisfy the properties in~\Cref{lemma:abstractionsound} while also triggering a conflict in the NRA plugin.

The refinement consists in conjoining to $\phi^\alpha$ the formula $\psi$ computed by~\Cref{pseudocode:refinement}. 
In a nutshell, this algorithm considers all subterms $f(t_1,\dots,t_r)$ of $\ell^\gamma$ 
where $f$ is a transcendental function, and adds (line~\ref{refinement:psi}) entailments 
stating that if the value assigned to $f_{t_j}^{\text{in}}$ lies in $I(t_j)$ for every $j \in [1..r]$, then the value assigned to $f_{(t_1,\dots,t_r)}^{\text{out}}$ must lie in $I(f(t_1,\dots,t_r))$, where $I(\rho)$ denotes the interval computed for subterm $\rho$ during the execution of~\Cref{pseudocode:consistency} on $\ell$ and $T$.

\begin{restatable}{lemma}{TheoremRefinementSpecification}\label{lemma:algo2-is-refinement}
    \Cref{pseudocode:refinement} 
    respects its specification. 
\end{restatable}

\begin{proof}[Proof idea]
    To prove that $T \not\models \psi$ one reasons by contradiction. Let $\tau
    \sim 0$ be the atomic formula in $\ell$. We assume that $T \models \psi$,
    and under this assumption show that $T \models \rho \in I(\rho)$ holds for
    each subterm~$\rho$ of the term $\tau$ ($\tau$ included). The induction relies
    on~\Cref{remark:initial-propagation} to handle function
    composition. This implies $T \models \tau \in I(\tau)$, 
    which in turn implies $b = T(\ell)$, 
    contradicting the hypotheses on the input of~\Cref{pseudocode:refinement}. 

    For the proof that every solution to $\phi$ can be extended into a solution to $\phi^\alpha \land \psi$, we rely on the fact that $\psi$ over-approximates the transcendental functions occurring in~$\tau$.
\end{proof}

After computing $\psi$, the TRA plugin identifies a conjunct $(\bigwedge_{\rho \in \mathcal{T}} i_{\rho} \leq f_{\rho}^{\text{in}} \leq j_{\rho}) \Rightarrow i_\eta \leq f_{\vec t}^\text{out} \leq j_\eta$ of this formula that is not satisfied by $T$. 
This reduces to evaluating inequalities between algebraic numbers, which can be done in polynomial time, as it corresponds to a query in the existential theory of the reals with a fixed number of variables~\cite{Renegar92}.

Let $x$ be the variable in the identified conjunct that was 
decided last in the trail $T$. The TRA plugin returns to MCSAT the formula $\psi$, 
to be conjoined to $\phi^\alpha$ (refining the abstraction), 
along with the variable $x$ to backtrack.
Before doing so, however, the TRA plugin performs some bookkeeping to ensure that the watchlist of literals waiting for the \mbox{$\delta$-consistency} check, and their values in the trail, will remain consistent with the trail after backtracking. More precisely, 
the plugin will re-add to the watchlist
every literal~$\ell$ on the trail that, after backtracking~$x$, will have at least one variable in $\allvars(\ell)$ become unassigned. Moreover, for any such literal, any trail assignment $\ell \mapsto \top^\delta$ is set back to $\ell \mapsto \top$.

\begin{remark}
    Let us assume that the variable $x$ is $f_\rho^{\text{in}}$ for some~$\rho$ (the case of $f_{\vec t}^{\text{out}}$ is similar).
    After backtracking $f_\rho^{\text{in}}$, unit propagation will add 
    either $\lnot (i_\rho \leq f_\rho^{\text{in}})$ or $\lnot (f_\rho^{\text{in}} \leq j_\rho)$ to the trail, ensuring that the antecedent of the identified conjunct is falsified. Then, should the NRA plugin be asked to decide a value for $f_\rho^{\text{in}}$, it will return a value different from the one assigned before backtracking.
\end{remark}

\subsection{Soundness of the procedure}

We show that ours is a $\delta$-satisfiability procedure:

\begin{restatable}[Soundness]{theorem}{TheoremSoundnessDelta}\label{theorem:soundness-delta}%
    Let $\phi$ be a TRA formula and $\delta > 0$ be a tolerance. If MCSAT with
    the TRA plugin in $\delta$-mode returns {\rm\texttt{SAT}},
    $\phi$~is satisfiable; if it returns {\rm\texttt{UNSAT}}, $\phi$ is
    unsatisfiable; if it returns {\rm\dsat}, $\phi$ is
    $\delta$-satisfiable.
\end{restatable}

\begin{proof}[Proof sketch]
    \textit{\texttt{SAT} case:} If MCSAT returns \texttt{SAT}, The trail $T$ assigns a value to all variables in $\phi^\alpha$, and every clause has a literal~$\ell$ 
    assigned to $\top$. When $\ell$ is \texttt{original}, 
    \Cref{remark:delta-consistency} gives $T \models \ell^\gamma$. 
    Hence, every \texttt{original} clause in $\phi$ has a literal satisfied by $T$, and $\phi$ is therefore satisfiable.

    \textit{\dsat case:} This case is similar to the previous one, but now there is a clause in $\phi^\alpha$ with no literal assigned to $\top$, but with a literal $\ell$ assigned to $\top^\delta$.
    \Cref{remark:delta-consistency} gives $T \models (\ell^\gamma)^\delta$, 
    and $\phi$ is found to be $\delta$-satisfiable.
    
    \textit{\texttt{UNSAT} case:} In this case, $\phi^\alpha$ is found to be unsatisfiable. By the second property in~\Cref{lemma:abstractionsound}, 
    which is preserved by \Cref{lemma:algo2-is-refinement}
    (\Cref{pseudocode:refinement}), $\phi$ is unsatisfiable as well.
\end{proof}

\subsection{More on \texorpdfstring{$\delta$}{delta}-satisfiability}\label{sec:-delta-vs-standard}

The notion of $\delta$-satisfiability procedure was introduced in~\cite{gao2012delta} to provide theoretical grounding for the use of numerical methods in decision procedures.
Different procedures may exploit the flexibility this notion offers in different ways. 
In our setting, 
we only employ $\delta$-weakening when checking the consistency of literals involving transcendental functions (\Cref{pseudocode:consistency}). 
This means, for example, that we never relax literals that are in \nra.

Note that even 
 the definition of \tra \mbox{$\delta$-consistency} check  (Def.~\ref{def:delta-consistency})
 allows some flexibility in the output.
For example, when the interval $[i,j]$ computed for a literal $\tau < 0$ is included in $(0,\delta)$, a $\delta$-consistency check may return either $\bot$ or $\top^\delta$. 
However, the table that we use in~\Cref{pseudocode:consistency}, which returns $\bot$ in this case, is the most precise possible. 
This guarantees that, within a single consistency check, an inconsistency output (that leans toward proving \unsat) is always preferred over a $\delta$-consistency one (which leans toward proving \dsat).

Our $\delta$-satisfiability procedure  naturally extends to a standard \sat/\unsat procedure. It suffices to iteratively run the $\delta$-satisfiability procedure starting with a large value of~$\delta$ (e.g.,~$\delta = 1$), halving~$\delta$ and restarting when finding \dsat. Crucially, restarting can be done while preserving the current abstraction $\phi^\alpha$, since the learned lemmas are independent of~$\delta$.
Decision heuristics (based, e.g., on caching~\cite{HIG25}), can further improve efficiency during restarts.
By never returning \dsat, \Cref{theorem:soundness-delta} 
guarantees that this \sat/\unsat mode is sound with respect to the standard 
notion of satisfiability.

\section{Specialized treatment of the sine and exponential functions}\label{section:special-treatment}

The previous section outlined a general procedure for \tras 
based solely on interval computability. Here, we show how additional assumptions can be used to refine this base procedure. We focus on the sine and exponential functions, though the approach extends 
in a modular way to other analytic functions.
Formally, we now consider the structure 
$(\R, 0, 1, \pi, e, +, \cdot, \sin, \exp, f_1,\dots,f_k, <, =)$,
where $f_1,\dots,f_k$ are interval computable functions.
(We recall that we see the constants $\pi$ and $e$ as interval-computable unary constant functions applied to the constant $1$.)

\subsection{Refining clause learning}

The main improvement\footnote{A further improvement, omitted for lack of space, consists of adding clauses (none marked as~\texttt{original}) 
to the initial abstraction $\phi^\alpha$, 
limiting the values of ``in'' and ``out'' variables. 
For example, we add constraints such as~$-1 \leq \sin_t^\text{out} \leq 1$, $1 + \exp_t^{\text{in}} \leq \exp_t^\text{out}$, and $3.14 < \pi_1^{\text{out}} < 3.15$.} 
to the procedure 
concerns refinements (i.e., the conjunctions added in line~\ref{refinement:psi} of~\Cref{pseudocode:refinement}). 
Let $h$ be either $\sin$ or $\exp$, and consider 
a conjunct of the form
\begin{equation}\label{eq:gamma}
    \gamma(h^{\text{in}}_t,h^\text{out}_t) \coloneqq (i_{\text{in}} \leq h^{\text{in}}_t \leq j_{\text{in}} \Rightarrow i_{\text{out}} \leq h^{\text{out}}_t \leq j_{\text{out}}),\,
\end{equation}
conjoined to $\psi$ at this line. Our goal is to strengthen $\psi$ by 
replacing $\gamma$ with a formula $\chi$ that entails it, while preserving the specification of~\Cref{pseudocode:refinement}.
To obtain a substantial strengthening of $\gamma$, we furthermore 
require $\chi$ to satisfy: 
\begin{description}
\item[\labeltext{P}{property:p}:] for every~$[a,b] \subseteq [\round{i_{\text{in}}}-1,\round{j_{\text{in}}}+1]$ 
with $b - a \leq j_{\text{in}} - i_{\text{in}}$, where $\round{k}$ stands for the integer closest to $k$, the possible values $h_t^{\text{out}}$ can take 
    in a solution to $a \leq h_t^{\text{in}} \leq b \land \chi$ all lie in an interval 
    of width at most $j_{\text{out}} - i_{\text{out}}$.
\end{description}
Property~\ref{property:p} guarantees that $\chi$ approximates $h$ well throughout an interval of width (at least)~$2$. 
We now define the formula $\chi$ for the sine and exponential functions.

\subsection{Definition of $\chi$ for the sine function}\label{subsec:chi-sine}

For the sine function, we define $\chi$ through Taylor approximations.
Let ${f \colon \R \to \R}$ be an infinitely differentiable function, 
and ${c \in \R}$. We recall that the Taylor expansion of $f$ centred at $c$ is the infinite sum
$\sum_{k=0}^\infty \frac{f^{(k)}(c)}{k!}(x-c)^k$, where
$f^{(k)}$ denotes the $k$th derivative of~$f$. The truncation to the
first $n$ terms of the sum is the \emph{$n$th Taylor polynomial of~$f$ centred at $c$}.

For $r \in \Q$ and $d \in \N$, define $({\rm T}^d_r\sin)(x)$ as the term 
$\sum_{k=0}^d \frac{\sin^{(k)}(r \cdot \pi)}{k!}(x- r \cdot \pi_1^{\text{out}})^k$.
It has the following properties:
\begin{enumerate}
    \item\label{prop:sin:1} If $r$ is of the form $2 \cdot j + \frac{1}{2}$ or $2 \cdot j + \frac{3}{2}$, for some $j \in \Z$, 
then $\sin^{(k)}(r \cdot \pi)$ is in $\{-1,0,1\}$. 
    \item\label{prop:sin:2} Assuming that $\pi_1^{\text{out}}$ is assigned the value $\pi$, then 
    \begin{align*}
        \textstyle
        ({\rm T}^{4 \cdot n+2}_{2 \cdot j + \frac{1}{2}}\sin)(x) &\leq \sin(x) \leq ({\rm T}^{4 \cdot n}_{2 \cdot j + \frac{1}{2}}\sin)(x), \text{ and}\\
        \textstyle ({\rm T}^{4 \cdot n}_{2 \cdot j + \frac{3}{2}}\sin)(x) &\leq \sin(x) \leq ({\rm T}^{4 \cdot n+2}_{2 \cdot j + \frac{3}{2}}\sin)(x), 
    \end{align*} for all $j \in \Z$, $n \in \N$, and $x \in \R$.
    \item\label{prop:sin:3} For every $r \in \Q$, $y = ({\rm T}^d_r\sin)(x)$ tends to the graph of $y = \sin(x)$ as $d \to \infty$ and $\pi_1^{\text{out}} \to \pi$.
\end{enumerate}

Consider~$\gamma(\sin_t^{\text{in}},\sin_t^{\text{out}})$ from~\Cref{eq:gamma}.
Let $p_\ell,p_u \in \Z$ be such that $[\round{i_{\text{in}}}-1,\round{j_{\text{in}}}+1] \subseteq [2 \cdot \pi \cdot p_\ell, 2 \cdot \pi \cdot p_u]$ and $p_u-p_\ell$ is minimal. 
For $\ell < u \in \Q$ with $\pi \in [\ell,u]$ and $n \in \N$, let $\chi_{[\ell,u]}^n(\sin_t^{\text{in}},\sin_t^{\text{out}},\pi_1^{\text{out}})$ be the formula given by:
\begin{align*}
    & \ell < \pi_1^{\text{out}} < u\\ 
    \land& \bigwedge\nolimits_{j = p_\ell}^{p_u} 
        \textstyle ({\rm T}^{4 \cdot n+2}_{2 \cdot j + \frac{1}{2}}\sin)(\sin_{t}^{\text{in}}) \leq \sin_{t}^{\text{out}} \leq ({\rm T}^{4 \cdot n}_{2 \cdot j + \frac{1}{2}}\sin)(\sin_{t}^{\text{in}})\\
    \land& \bigwedge\nolimits_{j = p_\ell}^{p_u}  
        \textstyle ({\rm T}^{4 \cdot n}_{2 \cdot j + \frac{3}{2}}\sin)(\sin_{t}^{\text{in}}) \leq \sin_{t}^{\text{out}} \leq ({\rm T}^{4 \cdot n+2}_{2 \cdot j + \frac{3}{2}}\sin)(\sin_{t}^{\text{in}}). 
\end{align*}
By Items~\ref{prop:sin:1} and~\ref{prop:sin:2}, $\chi_{[\ell,u]}^n$ is a satisfiable formula of NRA.
Moreover, Item~\ref{prop:sin:3} implies that, as $n$ increases and the interval $[\ell,u]$ is taken as a tighter approximation of $\pi$, the set of solutions to $\chi_{[\ell,u]}^n$ converges to that of the~\tra formula ${\pi_1^{\text{out}} = \pi \land \sin_t^{\text{out}} = \sin(\sin_t^{\text{in}})}$.
We may thus define $\chi \coloneqq \chi_{[\ell,u]}^n$ for $n$ sufficiently large and $[\ell,u]$ sufficiently tight rational approximation of $\pi$. This construction is effective:

\begin{restatable}[]{lemma}{LemmaSineChi}\label{lemma:sinechi}
    One can algorithmically compute $\ell,u$ and $n$ such that the formula $\chi_{[\ell,u]}^n$ entails $\gamma$ and satisfies Property~\ref{property:p}. Furthermore, substituting $\gamma$ with~$\chi_{[\ell,u]}^n$ in the output formula~$\psi$ of~\Cref{pseudocode:refinement} 
    preserves all properties of $\psi$ guaranteed by the specification of that algorithm. 
\end{restatable}

\subsection{Definition of $\chi$ for the exponential function}\label{subsec:special-treatment-exp}

To define $\chi$ for the exponential function, we can still rely on
Taylor approximations for lower bounds. For $c \in \Z$ and ${d \in \N}$, define
$(\rm{T}^d_c \exp)(x)$ as the term ${(e_1^{\text{out}})^c \cdot \sum_{k=0}^d
\frac{1}{k!}(x-c)^k}$. When $e_1^{\text{out}}$ is assigned the value $e$, 
this term corresponds to the \mbox{$d$-th} Taylor polynomial of $\exp$ centred at $c$, 
and we have ${(\rm{T}^{2n+1}_c \exp)(x) \leq \exp(x)}$ for every $n \in \N$, $c \in \Z$ and $x \in \R$.

Taylor approximations are inadequate for upper-bounding the
exponential function, as no finite polynomial approximation can serve as a
global upper bound for a super-polynomial function. 
For this purpose, we use Pad\'e approximants instead, 
which are generalizations of Taylor polynomials by rational functions~\cite{Padel}.

Given $c \in \Z$ and $d \in \N$, define the term
\begin{align*}
    ({\rm P}^d_c \exp)(x) \coloneqq (e_1^{\text{out}})^c \cdot \frac{Q_d(c,x)}{Q_d(x,c)}, 
\end{align*}
where $Q_d(x,y) \coloneqq \textstyle\sum_{k=0}^{d}\frac{\left(2d-k\right)!}{\left(d-k\right)!d!k!}\left(y-x\right)^{k}$.
When $e_1^{\text{out}}$ is assigned the value $e$, this term corresponds to the $d$-th Pad\'e approximant of $\exp$ centred at $c$, and has the following well-known properties
for every $n \in \N$ and $c \in \Z$:
\begin{enumerate}
    \item The polynomial $Q_{2n+1}(x,c)$ has a single root $r > c$, and
        $\exp(x) \leq ({\rm P}^{2n+1}_c \exp)(x)$ for every $x \in [c,r)$.\label{item:pole-pade}
    \item $\exp(x) \leq ({\rm P}^{2n}_c \exp)(x)$ for every $x \in (-\infty,c]$.\\
        (In this case, $Q_{2n}(x,c)$ has no roots.)
    \item $y = ({\rm P}^d_c \exp)(x)$ tends to the graph of $y = \exp(x)$ as $d \to \infty$ (and the root $r$ tends to infinity).
\end{enumerate}

Let~$\gamma(\exp_t^{\text{in}},\exp_t^{\text{out}})$ as in~\Cref{eq:gamma}, 
and $c \coloneqq \round{j_{\text{in}}}+1$.
For $\ell < u \in \Q$ with $e \in [\ell,u]$, and $n \in \N$, let $\chi_{[\ell,u]}^n$ be the formula in variables $\exp_t^{\text{in}}$, $\exp_t^{\text{out}}$ and $e_1^{\text{out}}$ given by:
{\begin{align*}
    & \ell < e_1^{\text{out}} < u \land (T_c^{2n+1}\exp)(\exp_t^{\text{in}}) \leq \exp_t^{\text{out}} \land {}\\ 
    & (\exp_t^{\text{in}} \leq c \Rightarrow \exp_t^{\text{out}} \leq (P_c^{2n}\exp)(\exp_t^{\text{in}})) \land {}\\
    &\exists r \big( Q_{2n+1}(r,c) = 0 
        \land  \big(c \leq \exp_t^{\text{in}} < r \Rightarrow {}\\ 
            &&\hspace{-3.1cm} \exp_t^{\text{out}} \leq (P_c^{2n+1}\exp)(\exp_t^{\text{in}})\big)\big).
\end{align*}}
One can show that~\Cref{lemma:sinechi} also holds for the formula $\chi_{[\ell,u]}^n$ above,
with respect to~$\gamma(\exp_t^{\text{in}},\exp_t^{\text{out}})$. 
It thus suffices to compute $\ell$, $u$ and $n$ as given in~\Cref{lemma:sinechi},
and define~$\chi \coloneqq \chi_{[\ell,u]}^n$. 

\section{Implementation and experimental evaluation}\label{sec:implementation}%
We implemented our procedure,
restricted to the structure $(\R, 0, 1, \pi, e, +, \cdot, \sin, \exp, <, =)$,
as an open-source plugin for \yices, relying on its MCSAT infrastructure
and NRA plugin.
We use the \tool{ARB}
library~\cite{Arb} for the arbitrary-precision interval arithmetic required by
the $\delta$-consistency check. Below, we refer to the resulting tool as
\yicestra.

For the benchmark, we use the same instances 
and time limit as~\cite{CimattiGIRS18}: 
2\,512 instances with a time limit of 1000 seconds, and no memory limit.
We refer to~\cite{CimattiGIRS18} for further details on the instances.
The results of our experiments are available in the accompanying  artifact~\cite{artifact}.

We performed two runs of \yicestra, one with~$\delta$ fixed at $10^{-3}$ and one in the \texttt{SAT}/\texttt{UNSAT} mode with decreasing~$\delta$.
The latter is denoted as \mbox{\yicestraDec} in plots.
We compared our tool with the SMT solvers \mathsat (v5.6.15) and
\cvc (v1.3.2), and with the $\delta$-satisfiability tool \dreal
(v4.21.06.2), with tolerances $\delta = 10^{-3}$ and $\delta = 10^{-9}$.
The two runs of \dreal performed very similarly, 
suggesting that varying~$\delta$ would not affect the results significantly.
\Cref{sec:related-work} gives more information on the methods these tools
implement. 

The results are presented using survival plots in~\Cref{figure:benchmark}, 
and show the competitiveness of our approach.  
Runs taking less than $0.1$ seconds are normalized to $0.1$ seconds for clarity.
With $\delta$ fixed, our tool reported 101 \texttt{SAT}, 954 \texttt{UNSAT} 
and 230 \dsat, outperforming all other tools in each category. 
The \texttt{SAT}/\texttt{UNSAT} mode classified a further 9 \texttt{SAT} and 8 \texttt{UNSAT} instances that were  reported as \dsat by the fixed $\delta$ run. In total, our tool classified as \texttt{SAT} or \texttt{UNSAT} 218 instances that remained unsolved by the other three tools.

Among the other tools, \cvc performs better on \sat instances and \mathsat on \unsat ones.
 Interestingly, although \yicestra solves significantly more \unsat instances than \mathsat within short time limits (e.g., less than 10s), 
 the latter tool almost catches up at 1000s, falling short of just 44 benchmarks.
 This seems to suggest that the linear lemmas \mathsat uses (see~\emph{incremental linearization} in \Cref{sec:related-work}) converge slower than our NRA lemmas, but might scale better. 
 Exploring how the degrees of the polynomial lemmas affects the performance and scalability of our procedure,
 and the development of hybrid approaches to automatically tune these degrees, are
 interesting directions for future work.


The last plot also counts \dsat instances as solved. We observe that the ratio of \dsat to solved instances differs markedly between \yicestra and \dreal. Specifically, with ${\delta = 10^{-3}}$, \dreal classifies 643 instances as \texttt{UNSAT} and 558 as \dsat; 
that is, 46\% of the instances solved by \dreal are classified as \dsat.
By contrast, \yicestra classifies only 18\% of the instances it solves as~\dsat.
Looking more closely at the \dsat instances, \yicestra 
returned \texttt{UNSAT} for 226 instances that \dreal classified as \dsat, 
while only 9 instances classified as \dsat by \yicestra were classified as \texttt{UNSAT} by \dreal. 
This suggests that, while \dsat might intuitively give the impression that a formula is  ``probably satisfiable'', in practice it often is not.

\begin{figure}
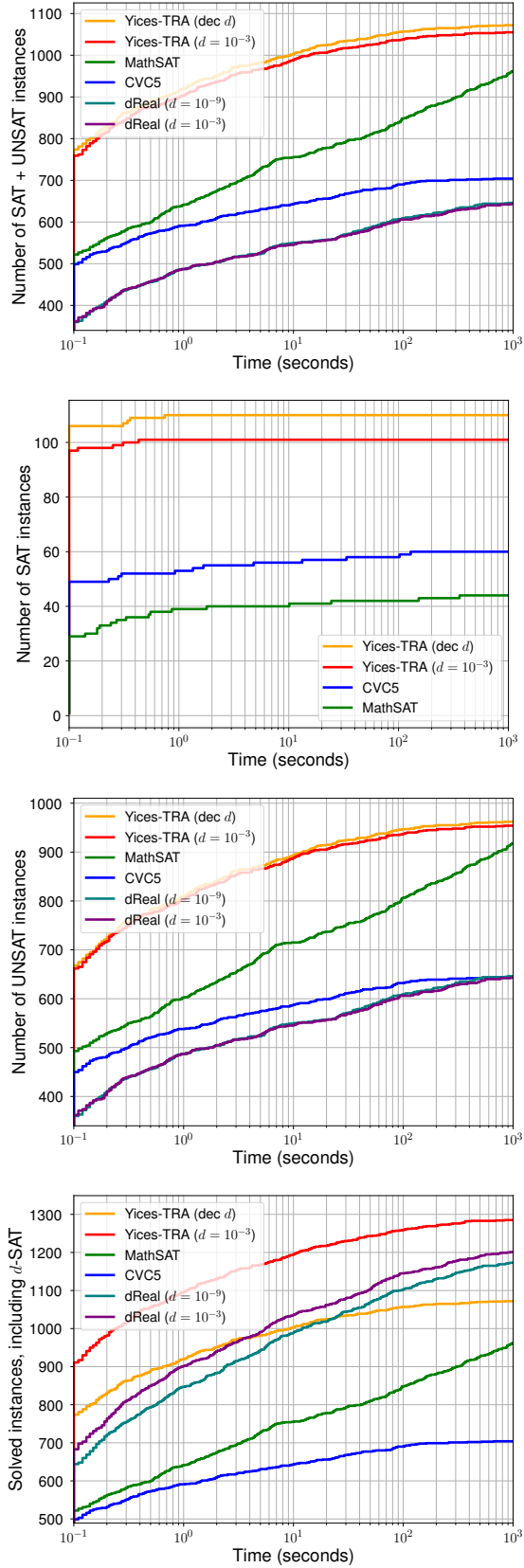

\centering
\scalebox{0.95}{%
\begin{tabular}{c}
\scalebox{0.41}{\input{Jupyter/solver-competition.pgf}}\\[2pt]

\scalebox{0.41}{\input{Jupyter/sat-only.pgf}}\\[2pt]

\scalebox{0.41}{\input{Jupyter/unsat-only.pgf}}\\[2pt]

\scalebox{0.41}{\input{Jupyter/delta-sat-only.pgf}}
\end{tabular}}

\vspace{10pt}
\caption{Cactus plots of the results of the benchmark.}\label{figure:benchmark}
\end{figure}









\section{Related Work}\label{sec:related-work}



\textit{Incremental Linearization.}
The idea of incremental linearization~\cite{CimattiGIRS18} is to lazily approximates all non-linear
functions with linear constraints, and solve the approximation with linear solvers. If the approximation is found unsatisfiable, so is the original formula. Otherwise, the model is either actual or spurious. 
Linear lemmas are added to exclude the spurious model. 
\mathsat and \cvc implement this procedure.
Ours is also an abstraction-refinement approach, 
but uses NRA lemmas with higher-degree polynomials. A benefit of this approach is that we can return irrational solutions, while linear solvers are limited to rational ones. Another distinction is that both 
\mathsat and \cvc are based on CDLC(T). We instead work in the MCSAT framework, which is known to perform particularly well with NRA, 
as it allows CAD computations to be performed locally near conflicting assignments.

\smallskip
\textit{Interval Constraint Propagation~(ICP).}
ICP is a branch-and-prune technique~\cite{benhamou2006continuous,ratschan2006efficient}. 
Given a starting box, is uses interval arithmetic to determine constraint signs over the box. If all constraint signs are constant, ICP returns whether the formula is satisfied. Else, it splits the box into smaller boxes and recurs. The results of the recursive calls are computed lazily: 
as long as one answers true, so does ICP.
It is implemented in \isat~\cite{iSAT3} and (under the $\delta$-satisfiability framework) in \dreal.
In our tool, we also use interval arithmetic, but its use is limited to computing the over-approximation of a function over a single point.

\smallskip
\textit{$\delta$-satisfiability.}
We have already introduced $\delta$-satisfiability
in~\Cref{section:tras}. The tool \dreal is tailored for this
relaxed notion of satisfiability, and can only return~\texttt{UNSAT} or
\texttt{$\delta$-SAT}. This differs from the $\delta$-consistency mode of our
tool, which can still answer \texttt{SAT}, and treats NRA
constraints exactly, since they are not affected by the TRA consistency check.
The \ksmt calculus~\cite{brausse2021ksmt}  combines  $\delta$-sat with linearization (however, the tool does not support transcendental functions).

\smallskip
\textit{Deductive methods.}  The \tool{MetiTarski}~\cite{MetiTarski} theorem prover can prove unsatisfiability for some formulas in \tras.
It replaces transcendental functions with user-defined lower and upper bounds, and relies on decision procedures for \nra to prove unsatisfiability.
Compared to our method, \tool{MetiTarski} is not able to prove satisfiability, and it only relies on the initial abstraction provided by the user.
%

\smallskip
\textit{Complete methods.} As stated in the introduction, even small fragments
of TRAs are undecidable. A notable exception is the result by Chen and
Xia~\cite{ChenX23}, who prove decidability for the univariate fragment with the sine
function, where atomic formulas are restricted to rational polynomial
inequalities in $x$ and $\sin(x)$ (disallowing complex terms inside sine).
In proving decidability, they devise a root isolation algorithm that could 
in principle be used in our procedure to tighten the feasibility sets 
computed by the NRA plugin. 
This, together with methods capable of proving satisfiability in the presence of transcendental solutions~\cite{SUPNLA,LippariniRatschanJAR}, as well as local search based approaches
~\cite{Atva22,LSnta-,LHIG25}, 
are  promising directions for further improving our procedure.

\vspace{\baselineskip}
\noindent\textbf{Acknowledgments.}
This work is part of a project co-funded by the European Union (GA 101154447)
and by MCIN/AEI (A CEX2024-001471-M and GA PID2022-138072OB-I00). Views and opinions expressed are
those of the authors only and do not necessarily reflect those of the European
Union or European Commission. Neither the European Union nor the granting
authority can be held responsible for them.


%

%
%
%

\bibliographystyle{alphaurl}
\bibliography{bibliography}

\newpage
\onecolumn
\appendix 



\section*{Additional material for Section~\ref{sec:mcsat-tra}}\label{appendix:sec2}

\subsubsection*{Formulas} Throughout the paper, we assume that all formulas are
quantifier-free and in \emph{conjunctive normal form (CNF)}. We recall that a
\emph{literal} is defined as an atomic formula or its negation. A \emph{clause} 
is a disjunction of literals, and a CNF formula is a conjunction of clauses.

\medskip
\subsubsection*{The CDCL Algorithm}
The main data structure of CDCL is given by the \emph{trail}, a sequence
${\ell_1 \mapsto b_1}$, \dots, $\ell_n \mapsto b_n$ assigning Booleans $b_i$ to
propositional literals $\ell_i$. The trail is \emph{consistent} whenever
$\bigwedge_{i=1}^n(\ell_i \leftrightarrow b_i)$ is satisfiable. A~consistent trail is
\emph{complete} whenever substituting each $\ell_i$ by $b_i$ in the input
formula yields, after the application of simple validities 
(e.g.,~$\top \lor \psi \leftrightarrow \top$), to the formula~$\top$. 
The algorithm reports SAT as soon as it constructs a complete trail.

Elements of the trail are annotated with information to implement a backtracking
mechanism: each element $\ell_i \mapsto b_i$ is annotated as either
\emph{decided} or \emph{propagated}, along with its \emph{decision level} (an
integer). \emph{Decided} literals are guessed by the algorithm and increment the
decision level. \emph{Propagated} literals are derived by unit propagation, and
inherit the current decision level. A clause is \emph{unit} if all but one of
its literals are assigned false in the trail. \emph{Unit propagation} consists
of assigning the remaining literal to true.

When a trail becomes inconsistent, CDCL analyses the annotations on the trail to
derive a \emph{learned clause} $\psi$ that explains the inconsistency
(a.k.a.~\emph{conflict}). By construction, this clause is implied by the input
formula~$\phi$, and the algorithm conjoins $\psi$ to~$\phi$. CDCL then backtracks
(non-chronologically) to a decision level at which the learned clause becomes
unit. If no such level exists, i.e., the conflict occurs at decision level~0,
the algorithm terminates and reports UNSAT.
\section*{Additional material for Section~\ref{sec:axioms-tra}}\label{appendix:sec3}

The goal of this appendix is to show that our procedure is sound (\Cref{theorem:soundness-delta}).
We start with some additional definitions and notation. 
Given a term $t$ (from a TRA) and a map $\nu$ assigning 
a value to each variable in $t$, we write $\llbracket t \rrbracket_\nu$ for the evaluation of $t$ under $\nu$. It is defined inductively as one expects: 
\begin{center}
  $\llbracket x \rrbracket_\nu = \nu(x)$, 
  $\llbracket t_1 + t_2 \rrbracket_\nu = \llbracket t_1 \rrbracket_\nu + \llbracket t_2 \rrbracket_\nu$,
  $\llbracket t_1 \cdot t_2 \rrbracket_\nu = \llbracket t_1 \rrbracket_\nu \cdot \llbracket t_2 \rrbracket_\nu$ and
  $\llbracket f(t_1,\dots,t_n) \rrbracket_\nu = f(\llbracket t_1 \rrbracket_\nu,\dots,\llbracket t_n \rrbracket_\nu)$.
\end{center}
When it is clear from the context, we omit the subscript~$\nu$ and simply write $\llbracket t\rrbracket$.

\begin{definition}[Canonical extension of a map]\label{def:canonical-extension}
  Consider a structure $(\R, 0, 1, +, \cdot, f_1,\dots,f_k, <, =)$, 
  where $f_1,\dots,f_k$ are transcendental functions. 
  Let~${\nu:X\mapsto \R}$ be a map from first-order variables to reals.
  The~\emph{canonical extension}~$\nu^{\text{ext}}$ of $\nu$ 
  is defined from $\nu$ as follows.
  For every $x \in X$, $\nu^{\text{ext}}(x) = \nu(x)$.
  For every term $f(t_1,\dots,t_n)$ from the structure, with $f$ among $f_1,\dots,f_k$, the map $\nu^{\text{ext}}$ contains assignments 
  $f_{t_i}^{{\text{in}}} \mapsto \llbracket t_i\rrbracket_\nu$, for all $i \in [1..n]$, and 
  $f^{\text{out}}_{(t_1,\dots,t_n)} \mapsto \llbracket f(t_1,\dots,t_n) \rrbracket_\nu$.

  (We are implicitly assuming that the variables $f_{t_i}^{{\text{in}}}$ and $f^{\text{out}}_{(t_1,\dots,t_n)}$ are not from $X$.)
\end{definition}

The following lemma is mostly helpful to clarify the computation performed in line~\ref{pseu:cons:K} of~\Cref{pseudocode:consistency}.

\begin{lemma}\label{aux-lemma:composition}
  Let $\tau$ be a term from the first-order theory of a TRA structure $(\R, 0, 1, +, \cdot, f_1,\dots,f_k, <, =)$, where $f_1,\dots,f_k$ are interval computable transcendental total functions. Let $\nu$ be a map from first-order variables to algebraic numbers, and $\delta > 0$ rational. 
  There is an algorithm that computes, for every subterm $\rho$ of $\tau$, 
  a non-empty interval $I(\rho) = [i_\rho,j_\rho]$ such that:
  \begin{itemize}
    \item $\llbracket \rho \rrbracket_\nu \in I(\rho)$.
    \item $j_\tau - i_\tau < \delta$.
    \item For every function application $f(t_1,\dots,t_n)$ in $\tau$ (where $f$ is among $+,\cdot,f_1,\dots,f_k$) the interval $I(f(t_1,\dots,t_n))$ is computed from the intervals $I(t_1),\dots,I(t_n)$ by applying the algorithm $I_f$ from \Cref{def:interval-comput}.
  \end{itemize}
\end{lemma}
\begin{proof}
  The algorithm is straightforward. 
  It iterates through $\epsilon = 1,\frac{1}{2},\frac{1}{4},\dots$. 
  For a given $\epsilon$, it computes non-empty algebraic intervals $I(x) = [\nu(x)-\epsilon,\nu(x)+\epsilon]$ for all variables $x$ in $\tau$, 
  and sets $I(0) = [-\epsilon,\epsilon]$ and $I(1) = [1-\epsilon,1+\epsilon]$. It then relies on the algorithms from \Cref{def:interval-comput} to compute all the remaining intervals $I(\rho)$. 
  If $j_\tau - i_\tau < \delta$ then the algorithm terminates. 
  Otherwise, it continues with the next $\epsilon$.
\end{proof}

\subsection{An invariant for the abstraction $\phi^\alpha$}

Throughout its execution, our procedure maintains and updates an NRA abstraction $\phi^\alpha$ of the input formula $\phi$. 
Towards a proof of soundness of the procedure, we show that this abstraction satisfies the following property throughout the procedure:

\medskip
\begin{description}
  \item[Invariant ($\ast$):]
    (1)~All variables in $\phi$ occur in $\phi^\alpha$ and
    (2)~for every solution $\nu$ to $\phi$, $\nu^{\text{ext}}$ is a solution to~$\phi^\alpha$.
\end{description}

\medskip
Let us prove that Invariant~($\ast$) is satisfied by the first abstraction $\phi^\alpha$. This is essentially~\Cref{lemma:abstractionsound}:

\LemmaAbstractionSound*
\begin{proof}
  We prove the stronger version of the lemma required 
  by Invariant~($\ast$), showing that for every solution $\nu$ to $\phi$, the canonical extension of $\nu$ is a solution to $\phi^\alpha$.

  Recall that $\phi^\alpha$ is obtained from $\phi$ by 
  iterative rewritings, starting from the outermost terms, of the form
  \begin{equation*}
    \phi \longrightarrow \phi\sub{f_{(t_1,\dots,t_r)}^{\text{out}}}{f(t_1,\dots,t_r)} \land \textstyle\bigwedge_{i=1}^r f_{t_i}^{\text{in}} = t_i\,.
  \end{equation*}
  Item~(1) from the lemma is thus trivial.

  For a proof of Item~(2), consider a solution~$\nu$ to $\phi$, 
  and take its canonical extension $\nu^{\text{ext}}$. 
  We verify that this extended map satisfies 
  $\phi^\alpha$. We can do so by induction on the number of rewriting steps performed to obtain $\phi^\alpha$ from $\phi$. 
  The base case is trivial since $\phi^\alpha = \phi$. 
  For the inductive step, 
  let $\psi$ be the formula obtained from $\phi$ after performing some rewriting steps, and consider 
  $
    \psi' \coloneqq {\psi\sub{f_{(t_1,\dots,t_r)}^{\text{out}}}{f(t_1,\dots,t_r)} \land \textstyle\bigwedge_{i=1}^r f_{t_i}^{\text{in}} = t_i}.
  $
  By the inductive hypothesis, $\nu^{\text{ext}}$ satisfies $\psi$. 
  We show that $\nu^{\text{ext}}$ also satisfies $\psi'$.
  Since $f(t_1,\dots,t_r)$ occurs in $\phi$, 
  we have:
  \begin{enumerate}[label=(\alph*)]
    \item $\nu^{\text{ext}}(f_{t_i}^{\text{in}})  = \llbracket t_i \rrbracket_\nu = \llbracket t_i \rrbracket_{\nu^{\text{ext}}}$, for all $i \in [1..r]$.
    \item $\nu^{\text{ext}}(f_{(t_1,\dots,t_r)}^{\text{out}}) = \llbracket f(t_1,\dots,t_r)\rrbracket_\nu = \llbracket f(t_1,\dots,t_r) \rrbracket_{\nu^{\text{ext}}}$.
  \end{enumerate} 
  From Item~(a), $\nu^{\text{ext}}$ satisfies the conjunct ${\bigwedge_{i=1}^r f_{t_i}^{\text{in}} = t_i}$. 
  Consider then an atomic formula $\tau \sim 0$ $({{\sim} \in \{<,=\}})$. 
  From Item~(b), $\llbracket \tau \rrbracket_{\nu^{\text{ext}}} = \llbracket \tau\sub{f_{(t_1,\dots,t_r)}^{\text{out}}}{f(t_1,\dots,t_r)}  \rrbracket_{\nu^{\text{ext}}}$, which in turn implies that
  $\nu^{\text{ext}}$ satisfies $\tau \sim 0$ if and only if 
  it also satisfies $\tau\sub{f_{(t_1,\dots,t_r)}^{\text{out}}}{f(t_1,\dots,t_r)} \sim 0$. We conclude that $\nu^{\text{ext}}$ satisfies $\psi'$.
  %
  %
\end{proof}




Throughout the execution of our procedure, the abstraction $\phi^\alpha$ is only updated by conjoining it to the formulas computed by \Cref{pseudocode:refinement}.
The following lemma shows that these formulas are satisfied by the canonical extension of every solution~to~$\phi$. 

\begin{lemma}
  \label{lemma:invariant-refinement}
  Let $\psi$ be the formula computed by \Cref{pseudocode:refinement} on input $(\ell,T)$, and let $\nu$ be a solution to $\phi$. Then, the canonical extension of $\nu$ is a solution to $\psi$.
\end{lemma}

\begin{proof}
  Let $\tau \sim 0$ be the atomic formula in $\ell$.
  Using the algorithm from~\Cref{aux-lemma:composition},
  \Cref{pseudocode:refinement} recomputes (line~\ref{refinement:Isubterms}) the intervals $I(\rho) = [i_\rho,j_\rho]$ for all
  subterms $\rho$ of $\tau^\gamma$ (including $\tau^\gamma$ itself).
  The formula $\psi$ is then defined as the conjunction of all implications
  of the form $\big((\bigwedge_{r=1}^n i_{t_r} \leq f_{t_r}^{\text{in}} \leq j_{t_r}) \Rightarrow i_\eta \leq f_{(t_1,\dots,t_n)}^\text{out} \leq j_\eta \big)$, 
  for every $\eta \coloneqq f(t_1,\dots,t_n)$ subterm of $\tau^\gamma$, with $f$  transcendental function (line~\ref{refinement:psi}).

  Consider now a solution $\nu$ to $\phi$, and let $\nu^{\text{ext}}$ be its canonical extension. Suppose $\nu^{\text{ext}}$ satisfies the antecedent of one of the implications in $\psi$, i.e., $\nu^{\text{ext}}(f_{t_r}^{\text{in}}) \in [i_{t_r},j_{t_r}]$ for all $r \in [1..n]$. 
  Equivalently, $\llbracket t_r \rrbracket_\nu \in I(t_r)$ for all $r \in [1..n]$.
  Form the last item in~\Cref{aux-lemma:composition}, $I(f(t_1,\dots,t_n))$ is computed from the intervals $I(t_1),\dots,I(t_n)$ by applying the algorithm $I_f$ from \Cref{def:interval-comput}. 
  From the property ``$f(B) \subseteq C$'' of the algorithm~$I_f$, stated in \Cref{def:interval-comput}, we conclude that  $\llbracket f(t_1,\dots,t_n) \rrbracket_\nu \in I(f(t_1,\dots,t_n))$. 
  By definition of the canonical extension, $\nu^{\text{ext}}(f_{(t_1,\dots,t_n)}^\text{out}) = \llbracket f(t_1,\dots,t_n) \rrbracket_\nu$.
  Therefore, $\nu^{\text{ext}}(f_{(t_1,\dots,t_n)}^\text{out}) \in I(f(t_1,\dots,t_n))$, and $\nu^{\text{ext}}$ satisfies the consequent of the implication. 
\end{proof}

\begin{lemma}
  \label{lemma:invariant-properties}
  Invariant~($\ast$) holds throughout the execution of the procedure.
\end{lemma}

\begin{proof}
  As already stated, $\phi^\alpha$ is only updated by conjoining formulas $\psi$ computed by \Cref{pseudocode:refinement}. 
  By \Cref{lemma:abstractionsound}, the invariant holds for the first abstraction.

  The first (syntactic) property of the invariant is unaffected by updates to $\phi^\alpha$, since the procedure only conjoins new clauses to it.
  The second property is ``closed under conjunctions'': for every $\nu$ solution to $\phi$, 
  if the canonical extension of $\nu$ is a solution to $\phi^\alpha$ and to $\psi$, then it is also a solution to $\phi^\alpha \land \psi$. 
  Therefore, the second property is preserved by~\Cref{lemma:invariant-refinement}.
\end{proof}

\subsection{Correctness of~\Cref{pseudocode:consistency} and \Cref{pseudocode:refinement}}

\TheoremConsistencySpecification*
\begin{proof}
  Let $(\ell,T)$ be an input literal and trail, with~$\allvars(\ell) \subseteq \dom(T)$.
  If $\ell^\gamma$ contains no transcendental function, then ${\ell^\gamma = \ell}$. That is, $\ell^\gamma$ is a
  formula from NRA. The NRA plugin already verified that $T(\ell)$ is consistent with its satisfaction, 
  and since the NRA plugin is sound by \Cref{assumption:nra}, it holds that 
  $T\models \ell$ if and only if $b=\top$ and hence $T\models \ell^\gamma$ if and only if $b=\top$. 
  
  Consider now the case where $\ell^\gamma$ is not in \nra. Let $\tau\sim 0$ the atomic formula in $\ell$.
  Using the algorithm from~\Cref{aux-lemma:composition}, 
  the algorithm computes a non-empty interval $[i,j]$ such that 
  $0 < j-i < \delta$ and $T\models i \leq \tau^\gamma \leq j$.

  The table in line~\ref{pseudo:returntable} perform a 
  complete case analysis on the possible values of $i$ and $j$, 
  with respect to their position with respect to the origin.
  We show that the algorithm is correct for $\ell$ of the form $\tau < 0$
  (and hence $\ell^\gamma$ equal to $\tau^\gamma <0$). 
  All other cases are similar.

  \begin{itemize}
    \item If $j < 0$ then $\tau^\gamma \leq j< 0$, 
    which implies $T\models \ell^\gamma$. Accordingly, the algorithm returns $\top$.
    \item If $i > 0$ then $\tau^\gamma \geq i> 0$, which implies $T\not\models \ell^\gamma$. Accordingly, the algorithm returns $\bot$.
    
    \item If $i = j = 0$ then $\tau^\gamma = 0$, which implies $T\not\models \ell^\gamma$. Accordingly, the algorithm returns $\bot$.
    \item If $0 = i < j$ then $\tau^\gamma \geq i = 0$, which implies $T\not\models \ell^\gamma$. Accordingly, the algorithm returns $\bot$.
    \item Finally, if $i < 0 \leq j$ then $\tau^\gamma \in [i,j]$ with $j-i < \delta$. Therefore, $-\delta < \tau^\gamma < \delta$. 
    This implies $T\models (\ell^\gamma)^\delta$, 
    and, accordingly, the algorithm returns $\top^\delta$.
    \qedhere
  \end{itemize}
\end{proof}

We now move to the proof of the correctness of \Cref{pseudocode:refinement}. 
We will first show it assuming that the formula $\phi^\alpha$ satisfies the invariant mentioned at the beginning of this section. 

\TheoremRefinementSpecification*
\begin{proof}
  Let $(\ell,T)$ be the input literal and trail such that $\allvars(\ell) \subseteq \dom(T)$.
  Let also $b$ be the output of \Cref{pseudocode:consistency} with input $(\ell,T)$ where
  $b \in \{\top,\bot\}$, and $T(\ell) \neq b$. 
  Recall that~\Cref{pseudocode:consistency} and~\Cref{pseudocode:refinement} 
  are called when the consistency check of the NRA plugin passes, i.e., $T$ satisfies all (NRA) literals in the trail.
  We have already proven that $\phi$ can be extended to a solution of $\phi^\alpha\land \psi$ when establishing~\Cref{lemma:invariant-properties}. 
  Therefore, in order to prove the lemma it suffices to show that $T \not\models \psi$.

  Assume~$\ell$ of the form $\tau < 0$ (the other cases are similar), 
  and let~$\nu \colon \allvars(\ell) \to \mathbb{R}$ 
  be the map defined as $\nu(x) \coloneq T(x)$ for all $x \in \allvars(\ell)$.
  
  Using the algorithm from~\Cref{aux-lemma:composition},
  \Cref{pseudocode:refinement} recomputes (line~\ref{refinement:Isubterms}) the intervals $I(\rho)$ for all
  subterms $\rho$ of $\tau^\gamma$ (including $\tau^\gamma$ itself).
  From $T(\ell)\neq b$ it follows that $\llbracket \tau \rrbracket_\nu \notin I(\tau^\gamma)$.
%
%
  Recall that the formula $\psi$ contains, 
  for every subterm~$\eta = f(t_1,\dots,t_n)$ of $\tau^\gamma$ with $f$ transcendental, 
  the formula $\left(\bigwedge_{i=1}^n f^{\text{in}}_{t_i}\in I(t_i)\Rightarrow f^{\text{out}}_{\vec{t}}\in I(f(t_1,\dots,t_n))\right)$ as a conjunct.

  For the sake of contradiction suppose that $T\models \psi$. 
  We show that then $\llbracket f^{\text{in}}_{t_i}\rrbracket_\nu \in I(t_i)$ 
  for every $\eta = f(t_1,\dots,t_n)$ as above. 
  The proof is by induction on the number of functions applications occurring in $t_i$,
  including additions~$+$ and multiplications~$\cdot$.
  Note that $t_i$ is a term from the original formula $\phi$. 

  \begin{description}
    \item[base case: $t_i$ is a variable $x$, the constant $0$, or the constant $1$.] 
      The latter two cases are trivial, since $0\in I(0)$ and $1\in I(1)$ by construction of the intervals. When $t_i$ is a variable $x$, 
      following~\Cref{remark:initial-propagation} we see that 
      $(f^{\text{in}}_x = x) \mapsto \top$ belongs to the trail. 
      Therefore, $\llbracket f^{\text{in}}_x \rrbracket_\nu  = \llbracket x \rrbracket_\nu$. By~\Cref{aux-lemma:composition}, $\llbracket x \rrbracket\in I(x)$, and so $\llbracket f^{\text{in}}_x \rrbracket \in I(x)$.
    \item[inductive case.]
      Consider $t_i = P(f_{i_1}(\vec t_1),\dots,f_{i_m}(\vec t_m))$, where $P$ is a polynomial, $f_{i_1},\dots,f_{i_m}$ are transcendental functions, 
      and $\vec{t_j} \coloneqq (t_{j1},\dots,t_{jn_j})$.
      By the induction hypothesis, for every $j \in [1..m]$ and for
      every $t \in \{t_{j1},\dots,t_{jn_j}\}$, 
      we have~$\llbracket (f_{i_j})^{\text{in}}_{t} \rrbracket \in I(t)$.
      So, as we are assuming $T\models\psi$, we have ${\llbracket (f_{i_j})^{\text{out}}_{\vec{t_j}}\rrbracket \in I(f_{i_j}(\vec{t_j}))}$, for every ${j \in [1..m]}$.
      By applying \Cref{aux-lemma:composition} on all additions and multiplications in $P$, $\llbracket P((f_{i_1})^{\text{out}}_{\vec{t_1}},\dots,(f_{i_m})^{\text{out}}_{\vec{t_m}}) \rrbracket\in I(P(f_{i_1}(\vec{t_1}),\dots,f_{i_m}(\vec{t_m}))) = I(t_i)$.
      Following~\Cref{remark:initial-propagation}, 
      the trail contains the assignment 
      $(f^{\text{in}}_{t_i} = P((f_{i_1})^{\text{out}}_{\vec{t_1}},\dots,(f_{i_m})^{\text{out}}_{\vec{t_m}})) \mapsto \top$. 
      Hence, $\llbracket f^{\text{in}}_{t_i} \rrbracket = \llbracket P((f_{i_1})^{\text{out}}_{\vec{t_1}},\dots,(f_{i_m})^{\text{out}}_{\vec{t_m}}) \rrbracket$, 
      and $\llbracket f^{\text{in}}_{t_i} \rrbracket \in I(t_i)$.
  \end{description}
  This concludes the proof that $\llbracket f^{\text{in}}_{t_i}\rrbracket_\nu \in I(t_i)$, for every $\eta = f(t_1,\dots,t_n)$ described above.
  Observe that, as we are assuming $T\models \psi$, this implies $\llbracket f_{\vec{t}}^{\text{out}} \rrbracket \in I(f(\vec t))$ for all the "out" variables occurring in $\allvars(\ell)$.

  Finally, the term $\tau$ is of the form $Q((f_{i_1})^{\text{out}}_{\vec{t_1}},\dots,(f_{i_k})^{\text{out}}_{\vec{t_k}})$, for some polynomial $Q$, 
  and we have $\llbracket (f_{i_j})^{\text{out}}_{\vec{t_j}} \rrbracket \in I(f_{i_j}(\vec{t_j}))$ for every $j \in [1..k]$. 
  Applying again \Cref{aux-lemma:composition} on all 
  addition and multiplication of $Q$, one sees that $\llbracket \tau \rrbracket_\nu \in I(\tau^\gamma)$; a contradiction.
  Therefore, $T\not\models\psi$, concluding the proof.
\end{proof}

We can now complete the proof of soundness of the procedure:

\TheoremSoundnessDelta*
\begin{proof}

  For the procedure to return \texttt{SAT} it must be the case
  that for every \texttt{original} 
  clause~$c$, there is a fully assigned literal $\ell_c$ from
  $\phi^\alpha$ that is true in the trail. Following~\Cref{lemma:algo1-is-delta-consistency-check}, 
  \Cref{pseudocode:consistency} must have
  returned $\top$ on $\ell_c$, as otherwise the algorithm would have either
  triggered a conflict or replaced $\top$ with $\top^\delta$ (see code before~\Cref{remark:delta-consistency}). 
  By~\Cref{remark:delta-consistency} (and~\Cref{lemma:algo1-is-delta-consistency-check}), $T\models \ell_c^\gamma$. 
  Since this holds for all literals $\ell_c$, the formula $\phi$ is satisfiable.

  As explained in \Cref{subsec:high-level-view}, 
  For the procedure to return \dsat it must be the case
  that for every \texttt{original} 
  clause~$c$, there is a fully assigned literal $\ell_c$ from
  $\phi^\alpha$ that is assigned $\top$ or $\top^\delta$ in the trail, 
  and moreover there is at least one \texttt{original} clause 
  with no literal assigned $\top$.
  Following the same argument as for the \sat case,
  clauses from $\phi$ in which the corresponding clause in $\phi^\alpha$ has at least one literal assigned $\top$ are satisfied. 
  Consider then a clause $c$ with a literal $\ell_c$ assigned $\top^\delta$.
  We show that $T \models (\ell_c^\gamma)^\delta$. 
  Let $\ell_c$ be of the form $\tau < 0$; the other cases are analogous. 
  Since $T$ assigns $\top^\delta$ to $\ell_c$,~\Cref{pseudocode:consistency} 
  outputs $\top^\delta$ on $(\ell_c,T)$.
  Since $[i,j]$ in line~\ref{pseu:cons:K} satisfy $j-i<\delta$, 
  following the table in line~\ref{pseudo:returntable} we see that 
  $i < 0 \leq j$.
  Since $j-i<\delta$, we then have $-\delta < i \leq j < \delta$. 
  This implies $m(\tau^\gamma)<\delta$, and so $\tau^\gamma < \delta$ holds.


  Finally, let us assume that the procedure returns \texttt{UNSAT}.
  This means that a conflict is triggered 
  when the trail has no backtrack points (the trail feature no decisions). 
  We claim that $\phi^\alpha$ is unsatisfiable, which
  by~\Cref{lemma:invariant-properties} 
  implies that $\phi$ is unsatisfiable as well. 
  There are a priori three cases for where the inconsistency is detected:
  (1) at the Boolean level, (2) by 
  the \nra plugin, (3) by the \tra plugin 
  (via~\Cref{pseudocode:consistency}). 
  In cases~(1) and~(2), unsatisfiability of $\phi^\alpha$ follows from the soundness of the \nra plugin (\Cref{assumption:nra}) and the fact that MCSAT extends the Boolean CDCL procedure. Case~(3) is instead not possible:
  since the trail has no decisions, 
  the inconsistency detected by~\Cref{pseudocode:consistency} would be on an input $(\ell,T)$ such that $\allvars(\ell) = \emptyset$. This implies that $\ell$ is a variable-free NRA literal. 
  However,~\Cref{pseudocode:consistency} simply returns $T(\ell)$ on such literals, 
  which does not trigger a conflict. 
\end{proof}

\section*{Additional material for \Cref{section:special-treatment}}\label{appendix:more-details-sin-exp}

We prove \Cref{lemma:sinechi} in both the cases of $\chi^n_{[\ell,u]}$ 
refining the lemmas for the sine and for the exponential functions.

\begin{remark}\label{remark:chi-monotone}
  The formula $\chi^n_{[\ell,u]}$ is ``monotone'' in the following sense: 
  For every $n' \geq n$ and every $[\ell',u'] \subseteq [\ell,u]$, the formula $\chi^{n'}_{[\ell',u']}$ 
  entails the formula $\chi^n_{[\ell,u]}$.
  This follows directly from the fact that increasing $n$ tightens 
  the Taylor and/or Padé approximations, while shrinking $[\ell,u]$ tightens 
  the approximation of $\pi$ (or of $e$). This has two consequences:
  \begin{itemize}
    \item If $\chi^n_{[\ell,u]}$ entails $\gamma$, then for every $n' \geq n$ and every $[\ell',u'] \subseteq [\ell,u]$, the formula $\chi^{n'}_{[\ell',u']}$ entails $\gamma$.
    \item If $\chi^n_{[\ell,u]}$ satisfies Property~\ref{property:p}, then for every $n' \geq n$ and every $[\ell',u'] \subseteq [\ell,u]$, the formula $\chi^{n'}_{[\ell',u']}$ satisfies Property~\ref{property:p}.
  \end{itemize}
\end{remark}

\LemmaSineChi*
We proceed first with the proof for the sine function.
\begin{proof}
  We first show that such $\ell$, $u$ and $n$ exist.
  By~\Cref{prop:sin:3} in \Cref{subsec:chi-sine}, as $n \to \infty$ and $[\ell,u] \to \{\pi\}$, 
  the solutions of $\chi^n_{[\ell,u]}$ converge to those of 
  $\pi_1^{\text{out}} = \pi \land \sin_t^{\text{out}} = \sin(\sin_t^{\text{in}})$.
  Since the image of the sine function evaluated on the input interval lies 
  strictly within the output interval of $\gamma$ (\Cref{eq:gamma}), any sufficiently 
  close approximation also satisfies $\gamma$. Then, there exist $n$ large enough 
  and $[\ell,u]$ tight enough around $\pi$ such that $\chi^n_{[\ell,u]}$ 
  entails $\gamma$. Similarly, the convergence guarantees that 
  Property~\ref{property:p} is satisfied for sufficiently large $n$ 
  and sufficiently tight $[\ell,u]$. This is essentially because the error of Taylor 
  approximations goes to $0$ as $n \to \infty$ and $[\ell,u] \to \{\pi\}$.

  Let us discuss now how to compute $\ell$, $u$ and $n$.
  First, let us recall that $\pi$ is computable: there is a Turing $A$ machine 
  that on input $i \in \N$, returns a pair $(\ell_i,u_i)$ or rationals such that 
  $\pi \in [\ell_i,u_i]$ and $u_i - \ell_i = 2^{-i}$. 
  The algorithm iterates over $n \in \N$, starting from $0$. 
  At each step, it computes $[\ell,u] = A(n)$, and checks 
  whether $\chi^n_{[\ell,u]}$ entails $\gamma$ and satisfies 
  Property~\ref{property:p}.
  Both checks are decidable, as $\chi^n_{[\ell,u]}$ and $\gamma$ are NRA formulas, 
  and Property~\ref{property:p} can also be stated as an NRA formula: 
  \begin{equation}\label{eq:property-p-in-nra}
    \begin{aligned}
      \forall a,b \Big( &(b-a \leq j_{\text{in}} - i_{\text{in}} \land 
        \round{i_{\text{in}}}-1 \leq a < b \leq \round{j_{\text{in}}}+1) \implies\\
      &\qquad\exists c,d : 0 < d-c \leq j_{\text{out}} - i_{\text{out}} \land \big(\forall \sin_t^{\text{in}},\sin_t^{\text{out}},\pi_1^{\text{out}} (a \leq \sin_t^{\text{in}} \leq b \land \chi^n_{[\ell,u]} \implies c \leq \sin_t^{\text{out}} \leq d)\big) \Big).
    \end{aligned}
  \end{equation}
  (Even though not needed for the lemma, let us remark that deciding this entailments can in fact be done in polynomial time, as they only involve a fixed number of variables~\cite{Renegar92}.)
  From the existence argument given at the beginning of the proof, and~\Cref{remark:chi-monotone}, 
  the two entailment will eventually be both true, and at that point the algorithm terminates.

  The last statement of the lemma follows directly from the fact that~$\chi^n_{[\ell,u]} \models \gamma$.
\end{proof}
Now we comment on some straightforward adaptations for the proof for the exponential function.
\begin{proof}
  The proof follows the same structure as the sine case, with two 
  differences. First, the Turing machine $A$ now computes pairs $(\ell_,u_i)$ that sandwich 
  $e$ rather than $\pi$. Second, the error 
  bound for Padé approximations requires avoiding poles inside the 
  interval of approximation. This is ensured by the choice 
  $c \coloneqq \round{j_{\text{in}}}+1$: as observed in \Cref{subsec:special-treatment-exp}, 
  this places the unique root of $Q_{2n+1}(\cdot,c)$ outside the region 
  of interest (\Cref{item:pole-pade} in \Cref{subsec:special-treatment-exp}).
\end{proof}

%
%
%
%
%
\end{document}